\RequirePackage{fix-cm}
\documentclass[smallextended]{svjour3}       
\smartqed  

\pdfoutput=1
\usepackage{float}
\usepackage{amsmath, amsfonts, amssymb}
\usepackage{graphicx}
\usepackage{graphics}
\usepackage{latexsym}
\usepackage{color}
\usepackage{tabularx,ragged2e,booktabs}
\usepackage[hypcap=false]{caption}
\RequirePackage{url}
\usepackage{soul}
\usepackage{hyperref}
\usepackage[title]{appendix}
\usepackage[sort]{natbib}
\setcitestyle{aysep={},yysep={;}}
\usepackage{soul}

\providecommand{\keywords}[1]
{
  \small	
  \textbf{\textit{Keywords:}} #1
}

\newtheorem{model}[theorem]{Model}

\newcommand{\bx}{\mbox{\boldmath{$x$}}}

\newcommand{\bff}{\mbox{\boldmath{$f$}}}

\newcommand{\bbF}{\mbox{\boldmath{$F$}}}

\begin{document}

\title{A mathematical model of the  COVID-19   pandemic dynamics
with dependent variable  infection rate
}
\subtitle{Application to the Republic of  Korea}

\titlerunning{Mathematical model of  COVID-19}        

\author{Aycil Cesmelioglu \and Kenneth L. Kuttler \and Meir Shillor \and Anna M. Spagnuolo
}

\authorrunning{A. Cesmelioglu \and K. L. Kuttler \and M. Shillor \and A. M. Spagnuolo} 

\institute{Aycil Cesmelioglu \at
              Department of Mathematics and Statistics, Oakland University, Rochester, MI, USA
              \email{cesmelio@oakland.edu}           
              \emph{ORCID ID: 0000-0001-8057-6349} 
           \and
           Kenneth L. Kuttler \at
              retired, USA
              \email{klkuttler@gmail.com}   
           \and 
           Meir Shillor \at
              Department of Mathematics and Statistics, Oakland University, Rochester, MI, USA
              \email{shillor@oakland.edu}   
              \emph{ORCID ID: 0000-0001-6811-9524}
           \and 
           Anna M. Spagnuolo \at
              Department of Mathematics and Statistics, Oakland University, Rochester, MI, USA
              \email{spagnuolo@oakland.edu}   
              \emph{ORCID ID: 0000-0003-3039-0970}
}

\date{Received: date / Accepted: date}

\maketitle

\begin{abstract}
This work constructs, analyzes, and simulates a new compartmental  SEIR-type model for the dynamics 
and potential control of the current COVID-19 pandemic. The novelty in this work is two-fold.
First, the population is divided according to  its compliance with disease control directives 
(lockdown, shelter-in-place, masks/face coverings, physical distancing, etc.) into those 
who fully comply and those who follow  the directives partially, or are 
necessarily mobile (such as medical staff). This split,  indirectly, reflects on the
quality and consistency of these measures. This allows the assessment of the overall effectiveness 
of the control measures and the impact of their relaxing or tightening on the disease spread. Second,
the adequate contact rate, which directly affects the infection rate,  is one of the model 
unknowns, as it keeps track of the changes in the population behavior and the effectiveness of 
various disease treatment modalities  via a differential inclusion. 
Existence, uniqueness and positivity results are proved using a nonstandard convex analysis based approach.
   As a case study, the pandemic outbreak in the Republic of Korea (South Korea) is 
simulated. The model parameters were found by minimizing the deviation of the model 
prediction from the reported data over the first  $100$ days of the pandemic in South Korea.
The simulations  show that  the model captures accurately the  pandemic dynamics in the 
subsequent $75$ days,  which provides confidence in the model predictions and its future use.
In particular, the model predicts that about $40\%$ of the 
 infections were not documented, which implies that asymptomatic infections contribute silently 
 but substantially to the spread of the disease indicating that more widespread asymptomatic 
 testing is necessary.

\keywords{COVID-19 \and SARS-CoV-2 \and compartmental continuous model \and  nonlinear fitting \and simulations}
\subclass{MSC 92D30  \and 92B05 \and 92B99 \and 34A60 \and 34F05 }
\end{abstract}

\section{Introduction}
\label{intro}


The novel coronavirus SARS-CoV-2 emerged in Wuhan, China, in December 2019 as a 
mutation of the Severe Acute Respiratory Syndrome Coronavirus, SARS-CoV.  The first group 
of COVID-19 patients, reported in Wuhan 
exhibited flu-like symptoms resulting in the serious cases in pneumonia and death.  There is mounting 
evidence \citep{WHO,Marshall2020, Huang20} and the references therein,
that the virus affects blood vessels, and therefore, the body organs, which may cause 
long term health complications.  Currently, the pandemic affects all parts of the world, and 
can be found in over 200 countries and territories
\citep{WHO, Worldometer, JHU}. The World Health Organization (WHO) declared COVID-19 a pandemic on 
11 March 2020, which caused the introduction of emergency measures: large scale closing of borders, 
and lock-down of countries, states, regions, cities, and communities, as well as closing of schools 
and universities. As of 8 Aug.  2020, globally, there are over 14 million 
confirmed cases, over 600,000 deaths and over eight million  recovered   \citep{Worldometer, JHU}.
Thus, the number of COVID-19  infections 
by far exceed  the number of SARS or the Middle East Respiratory Syndrome (MERS) infections. 
The SARS-CoV-2 human-to-human virus transmission is mainly via airborne fluid droplets, 
especially among those who are in close proximity.  Transmission from infected patients to 
healthcare personnel has been often observed. 

Because of the global impact of the pandemic,  many researchers  
have been engaged in mathematical and statistical modeling of various aspects 
of the disease dynamics, aiming at predicting global and national trends, as 
well as specific regional behaviors of the virus spread. These are meant to help 
policymakers with hospital and emergency services preparedness and utilization; 
national and regional policy response decisions; business plans; implementation, relaxation 
and evaluation of various control measures; organization of large scale vaccination;
and scientific understanding of the many facets of the
disease dynamics. These efforts are reflected in the  large number of 
relevant publications listed in, see, e.g., 
\citep{CDCresearch,WHO,Wikipedia}, and many recent publications such as 
\citep{Hou20, Roumen, Eikenberry2020, FP20, Garba2020, GBB20}. An interesting mathematical model 
dealing with some of the social implications of the pandemic is  \cite{JP20}.
 Since this work is concerned with simulations until 8 Aug.  2020, we do not consider 
new virus strains, or the global vaccination efforts that are underway.

We construct a new model for the COVID-19 pandemic to help researchers and policymakers 
evaluate the effectiveness and effects of various intervention and mitigation strategies, 
in real-time. It is based on the ideas underlying the MERS model in \citep{ALSS16, AS18, A17}.  
However, our preliminary simulations of the COVID-19 pandemic using this MERS model produced 
unsatisfactory results. One of the main reasons is related to the  
government-mandated measures (how strict they were, how well they were communicated and enforced) and 
how well the population complied with them. Therefore, we distinguish
those who follow properly the directives (wear masks in public, keep 
an appropriate distance from other people and frequently wash hands and/or use disinfectants) and those who do not 
or only partially do so. The latter group includes also health care personnel who have to be in contact with 
sick people, retirement homes' personnel, and other essential workers. 
To account for this, we introduce a new time-dependent parameter $\theta$ that represents the division 
of the exposed and infected subpopulations into those who fully comply  and those who comply 
only partially. This parameter measures both the effectiveness of the directives and the population's compliance 
and allows to study the effects of these control measures and the possible impact 
of relaxing or tightening them on the disease spread. A related parameter in a simpler model was 
introduced in \citep{Roumen, Garba2020}, where the focus was on estimating the size of the 
exposed but asymptomatic population. Our model provides this information, too.
 
The second novelty  is in considering the `adequate contact parameter' $\beta$, \citep{hethcote00}, 
as one of the model dependent variables, instead of the usual assumption that it is a fixed constant, 
or a function of time, see e.g.,  (\cite{Greenhalgh95, Greenhalgh17, Thieme02, ONeill97}) and the references therein.  
The case when $\beta$ depends on the number of recovered $R$ 
can be found in \cite{Bobko2020}, where it is shown how the choice of $\beta(R)$ affects the stability of the 
equilibrium points in a `simple' SIR model. 
Here, instead of having to modify and fit $\beta$ as the pandemic progressed, 
we consider it as a dependent variable, since it is known that  $\beta$ changes with changing 
population behavior, (\cite{thieme03, Wiki-contactrate}),  with tightened/relaxed control measures, 
seasonal changes, and availability of more information on the spread of COVID-19 and its effects on 
human health. The change in $\beta$ can also be caused by the introduction of more 
effective face masks, and  possible mutations. The latter will likely involve randomness that  
we would like to study in the future.
Furthermore, in simple SEIR models, there is a saturation phenomenon, that is, a decline in infectivity unrelated 
to the `herd immunity,' see e.g., \citep{tkolokol20, ZZ03, HM93} and references therein. By choosing $\beta$ to be 
one of the model unknowns, we keep track of the intrinsic changes of the contact parameter, which 
includes the virus `infectiveness,' and also the changes in human behavior. Since $1/\beta$ is the average 
number of contacts needed to infect a healthy person, it has a restricted range, $0< \beta\leq 1$. This
restriction is modeled by a differential inclusion. For the sake of simplicity, we assume
that $\beta$ grows (linearly) with the infected and decays (linearly) with the recovered. This linearity 
assumption is \underline{{\it ad-hoc}} and is not based on any deeper epidemiological considerations.
Therefore, it is of considerable interest to explore more appropriate forms of the inclusion in the future.
Choosing $\beta$ as an unknown that is described by a differential inclusion captures the evolution of 
the pandemic more naturally and is expected to make the model more accurate.  However, our simulations for South Korea shows that  $\beta$ varies very little over the 175 days of simulations,
so taking it as a fixed parameter, in this case, leads to very similar results. This is likely due to 
their successful COVID-19 containment and mitigation strategies. Nevertheless, long-time simulations 
(1,000 days) show that $\beta$ changes to a lower value, which indicates that it is affected by the process.
A third novelty is the theoretical introduction of randomness into the model parameters. 
This wil be used in future works to study the sensitivity of the model to certain parameters and
 thus allowing for better predictions, as well as the introduction of virus mutations.

Since the model includes a differential inclusion, we establish the existence of 
the unique solution for the initial value problem,  based on a theorem 
in abstract Hilbert spaces, \citep{brezis}.   Moreover, the introduction of randomness
into the model parameters raises the question of the measurability of the solutions
with respect to the random variables. These mathematical issues are discussed in some detail
in Sections \ref{sec-exist} and \ref{sec-random}.

An algorithm for the model simulations was constructed and implemented in MATLAB. 
Several simulations that show the predictive power of the model are presented  
in the context of the pandemic dynamics in the Republic of  Korea. South Korea is chosen for the case study because it is one of the first countries to go through the disease cycle and the information provided by the government is very reliable.  In the simulations, we use some of the published data to `train' the model by using an optimization routine in MATLAB which finds the model parameters that provide a `best $\ell^1$ fit.'  We present our simulated 
model predictions together with the data, one part of which (the first 100 days) was 
used in the optimization while the remaining 75  days depict how close the model predictions 
and the data are. The simulations of 175 days indicate that the model captures well 
the disease dynamics in South Korea.

Following this introduction, the mathematical model for the pandemic is constructed
in Section \ref{sec-model}. It consists of a coupled  system of seven nonlinear 
ordinary differential equations and a differential inclusion for the contact variable $\beta$. The existence of the
unique solution to the model is  established in  Section \ref{sec-exist}.
Section \ref{sec-random}  describes very briefly the addition of randomness
and the measurability of the solutions with random parameters.
The stability of the two equilibrium states, the disease-free equilibrium
(DFE) and the endemic equilibrium (EE), is presented in Section \ref{sec-stability}, 
based on an expression for the system's Jacobian derived in the Appendix. 
Section \ref{sec-sims} outlines the algorithm used 
in the numerical simulations. Section \ref{sim-SK} reports the results of the 
simulations of the COVID-19 dynamics in South Korea. The baseline simulations are in
Subsection \ref{SKbaseline}, where the model predictions are compared to the data. 
This section also provides additional information about the various subpopulations,
which is difficult to obtain in the field, such as those who carry the virus,
can infect others, but are undocumented and asymptomatic,  
and the variation in $\beta$ as well as the stability of the DFE and EE.
The effectiveness of the control measures and their connection to $\theta$ is studied in
Subsection \ref{SK-thetaImp}, while Subsection \ref{SK-DR} presents the 
resulting case fatality and the infection fatality rates.
 Conclusions,  unresolved issues and future work can be found in Section \ref{sec-con}.

\section{The Model}
\label{sec-model}
This section  presents  a mathematical model for the dynamics of the COVID-19 pandemic.  
It is based on the ideas that led to the MERS model in \citep{ALSS16, AS18, A17}.  However, 
there are significant differences between the compartmental structures and the resulting set 
of equations of the two models. In particular, this model separates 
the subpopulations of those who comply with the disease control directives  and those 
who choose not to or partially comply. 

The model describes whole populations and assumes that they are large enough to 
justify continuous dependent variables, thus, we use ordinary differential equations 
(ODEs). When the geographical distribution  of the disease is important, or
when the density and culture of the population vary in  different parts of the 
country, it must be modified and partial differential equations (PDEs) need to be 
used instead. However, using PDEs considerably increases the mathematical 
complexity of the model and is beyond the scope of this work. 

The model assumes that there may be various disease control measures such as
voluntary or mandatory isolation, shelter-in-place directives, movement controls, closure of
various public spaces and places,
physical-distancing, and that a portion of the population practices these measures with varying degree. 
It describes the dynamics of seven subpopulations: 
{\it susceptibles} $S$; {\it  asymptomatic (exposed) fully compliant} $E_{fc}$;  
{\it  asymptomatic (exposed) partially compliant or  noncompliant} $E_{pc}$;   
{\it infected fully compliant}, $I_{fc}$; {\it partially compliant or  noncompliant infected} $I_{pc}$;
{\it hospitalized} $H$;  and {\it recovered} $R$. Unlike the usual SEIR models where those in 
the $E$ compartment(s) are latent and cannot infect others, here the term `exposed' also includes 
those who carry the virus and can infect others, but do not exhibit any symptoms. This compartment is of considerable 
interest in COVID-19 since a portion of the infections were caused by individuals in $E$. 
For the sake of simplicity, we refer below to partially compliant or 
noncompliant as partially compliant. In addition, the model includes a differential inclusion for the infection 
rate function $\beta$. The populations and $\beta$ are functions of time, which is measured in days.  
The compartmental structure of the model  is depicted in Figure~\ref{fig:compartmentalchart}.

\vskip2pt
 \hskip0.5cm
\unitlength 0.27mm
\begin{picture}(60,120)
\put( 0,-10){\framebox(40,40)}
\put(-28,16){$p_{S}$}
\put(-34,7){\vector(1,0){30}}
\put(16,6){$S$}
\put(6,-40){$\mu $}
\put(10,-15){\vector(0,-1){15}}
\put( 80,-54){\framebox(40,40)}
\put( 80,36){\framebox(40,40)}
\put(45,10){\vector(1,1){30}}
\put(92,52){$E_{fc}$}
\put(92,-36){$E_{pc}$}
\put(48,34){$\theta\, \Gamma$}
\put(18,-24){$(1-\theta)\,\Gamma$}
\put(45,10){\vector(1,-1){30}}
\put(120,-66){$\mu $} 
\put(120,22){$\mu $} 
\put(44,66){$p_{Efc}$}
\put(44,-56){$p_{Epc}$}
\put(115,-58){\vector(0,-1){15}}
\put(48,58){\vector(1,0){28}}
\put(48,-44){\vector(1,0){28}}
\put(115,30){\vector(0,-1){15}}
\put(100,12){\vector(0,1){22}}
\put(100,12){\vector(0,-1){22}}
\put(80,10){$\gamma_{E}$}
\put(100, 95){\line(1,0){260}}
\put(100, 76){\vector(0,1){19}}
\put(100, -82){\line(1,0){260}}
\put(100, -58){\vector(0,-1){24}}
\put( 185,-54){\framebox(40,40)}
\put( 185,36){\framebox(40,40)}
\put(197,52){$I_{fc}$}
\put(197,-36){$I_{pc}$}
\put(143,58){$\gamma_{fc}$}
\put(151,-26){$\gamma_{pc}$}
\put(164,2){$\gamma^{+}$}
\put(154,26){$\gamma^{-}$}
\put(122,42){\vector(1,-1){60}}
\put(122,-20){\vector(1,1){60}}
\put(220,-66){$\mu+d_{I} $} 
\put(205,22){$\mu+d_{I} $} 
\put(153,78){$p_{Ifc}$}
\put(153,-58){$p_{Ipc}$}
\put(215,-58){\vector(0,-1){15}}
\put(151,70){\vector(1,0){27}}
\put(151,-46){\vector(1,0){27}}
\put(125,52){\vector(1,0){55}}
\put(125,-32){\vector(1,0){55}}
\put(200,30){\vector(0,-1){15}}
\put(230,52){\vector(1,-1){42}}
\put(230,-32){\vector(1,1){42}}
\put(249,38){$\delta_{}$}
\put(249,-28){$\delta_{}$}
\put(203, 76){\vector(0,1){19}}
\put(209,84){$\sigma_{I}$}
\put(200, -58){\vector(0,-1){24}}
\put(194,12){\vector(0,1){22}}
\put(194,12){\vector(0,-1){22}}
\put(195,0){$\gamma_{I}$}
\put( 275,-10){\framebox(40,40)}
\put(285,6){$H$}
\put(299,46){$p_{H}$}
\put(295,58){\vector(0,-1){25}}
\put(319,7){\vector(1,0){28}}
\put(295,-12){\vector(0,-1){15}}
\put(301,-20){$\mu+d_{H} $}
\put( 350,-10){\framebox(40,40)}
\put(360,6){$R$}
\put(320,16){$\sigma_{H} $}
\put(78,85){$\sigma_{E}$}
\put(78,-70){$\sigma_{E}$}
\put(182,-70){$\sigma_{I}$}
\put(380,-20){$\mu $}
\put(384,46){$p_{R}$}
\put(380,58){\vector(0,-1){25}}
\put(375,-14){\vector(0,-1){15}}
\put(360,95){\vector(0,-1){62}}
\put(360,-82){\vector(0,1){70}}
\end{picture}
\vskip50pt
\begin{figure}[ht]
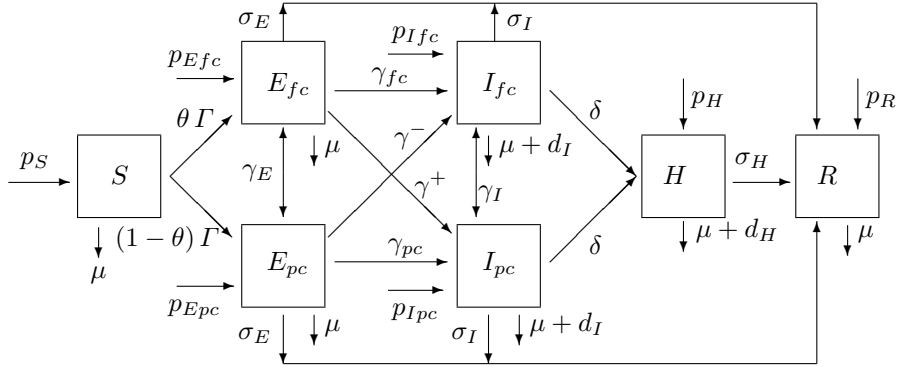

   \centering
   \caption{\small Compartmental structure and  flow chart for the COVID-19  model; 
   $\mu$ is the natural death rate; $p_{S}, \dots, p_{R}$ are influxes of individuals
   from outside; $\gamma_{fc}, \dots, \gamma_{I}$,  
   are infection rates; $\delta$ is the hospitalization rate; $\sigma_{E}, \sigma_I, \sigma_{H}$ are recovery rates; 
   and $d_I, d_{H}$ are the disease death rates; 
   $\Gamma$ is  the {\it force of infection} given in (\ref{eqn29}).}
   \label{fig:compartmentalchart}
\end{figure}

The susceptible subpopulation $S(t)$ consists of those who are healthy (only in terms of 
the COVID-19, while they may have other health conditions) and can become sick, or do not belong to  
any of the other groups.  The subpopulations $E_{fc}(t)$ and $E_{pc}(t)$ denote the current numbers 
of  those exposed to COVID-19 who fully comply with the disease control directives and those 
who do not or do so partially, respectively. The subpopulation $E_{pc}(t)$ includes also those who because of work
or other circumstances cannot fully follow the shelter-in-place and the other directives. Furthermore, 
the subpopulations $E_{fc}(t)$ and $E_{pc}(t)$ consist of: asymptomatics -- those who do not become 
sick; pre-symptomatics -- those who in 5-14 days develop clinical symptoms; and those who 
do not show symptoms or are mildly symptomatic but are not tested and documented.  
It is established that both subpopulations carry the virus and can infect susceptibles
 and as noted above, they include the latent individuals who cannot infect others and
also those who can.

The subpopulations $I_{fc}(t)$ and $I_{pc}(t)$  consist of documented infected individuals 
who are fully compliant or partially compliant, respectively, with no, mild, or medium 
clinical symptoms.  $H(t)$ denotes those whose symptoms are severe or 
critical and are hospitalized. 
Finally, $R(t)$ denotes the individuals who recovered 
from COVID-19 and as of Aug.  2020, there are approximately 9 million such individuals in the world. 
We assume that recovered individuals have temporal immunity and cannot become infected again. 
Moreover, at this stage it is not known how long this immunity lasts. 
If evidence of reinfection emerges, it is straightforward 
to modify the model to add reinfection pathways in the compartmental structure, 
possibly with delays.

 Next, we define the parameters of the model. 
The total living population at time $t$ is given by 
\[
N(t)=S(t)+E_{fc}(t)+E_{pc}(t)+I_{fc}(t)+I_{pc}(t)+H(t)+R(t).
\]

We denote by $p_{S}(t)$ the number of susceptible individuals that are added each 
day by birth or by travel from outside of the region of interest. 
We let $p_{Efc}(t)$ and $p_{Epc}(t)$
be the per day influx  of fully compliant and partially compliant
exposed, respectively, and note that this influx (mostly by air travel) seems 
to be one of the main reasons why COVID-19 spread so quickly around the globe
(see, e.g., \citep{virusspread}). 
Next, $p_{Ifc}(t)$ and $p_{Ipc}(t)$ denote the per day influxes of fully compliant infected and partially 
compliant infected populations, respectively. Then,  $p_{H}(t)$ is the number of those who arrive 
sick on day $t$ and need hospitalization, and $p_{R}(t)$ is the influx of recovered individuals.

The `natural' death rate coefficient (in the absence of COVID-19) of the population, 
$\mu\, (1/day)$, can easily be obtained from the demographic information of the country 
and can be considered as given.

Next, the {\it contact rate} (or the infection rate coefficient) $\beta(t)\, (1/day)$ 
at time $t$ is  the average number of contacts sufficient for transmission. The rates of infection of a 
susceptible by a fully or partially compliant asymptomatic, a  fully or partially compliant infected, 
or a hospitalized results in infection, respectively, are  
 \[
 \epsilon_{Efc}\beta  ,\quad \epsilon_{Epc}\beta,\quad \epsilon_{Ifc}\beta,\quad
 \epsilon_{Ipc}\beta,\quad \epsilon_{H}\beta,
 \]
$(1/day)$. The $\epsilon$'s are nonnegative dimensionless infection rate modification 
constants. When the isolation is very effective, the first and the last three are
likely to be small, since contacts are discouraged, and $\epsilon_{H}$ should be 
kept small in every hospital environment. 

The rate $\Gamma$, the so-called {\it force of infection}, is given by
\[
\Gamma=\frac{\beta}N (\epsilon_{Efc} E_{fc} +\epsilon_{Epc} E_{pc}+
\epsilon_{Ifc}I_{fc} + \epsilon_{Ipc} I_{pc} + \epsilon_H H),
\]
and measures the rate at which susceptibles get infected. We note that to model the case
when $E$ represents only those that are latent and cannot infect others, one has to set 
$\epsilon_{Efc}=\epsilon_{Epc}=0$.
The rate, per day,  at which the fraction $\theta $ of the susceptibles  
becomes exposed fully compliant is $\theta \Gamma S$, 
and the rest becomes exposed partially compliant at the rate $(1-\theta)\Gamma S$.
Moreover, we allow fully compliant exposed to become 
partially compliant exposed and vice versa, with rate constant $\gamma_{E}$. This may be 
triggered by a change in their views or their economic situation. The rate constant 
$\gamma_I$ has a similar interpretation. The 
parameters $\gamma_{fc}, \gamma_{pc}\ (1/day)$ denote the rate constants of 
development of clinical symptoms in fully compliant and partially compliant exposed 
individuals, while $\gamma^{-}, \gamma^{+}$ are the rate constants of clinical symptoms
in fully compliant who become partially compliant and infected (hopefully small), 
and those partially compliant 
exposed who become fully compliant infected. The model assumes that there is no 
disease-induced death in the exposed populations. Next, the rate constant  at which 
the fully compliant infected and the partially compliant infected individuals 
need hospitalization is  assumed to be the same, $\delta$. The additional 
disease-induced death for the infected and  hospitalized are $d_{I}$ and $d_{H}$, respectively. 
To complete the flow chart, we denote by $\sigma_{E}, 
\sigma_{I}$ and $\sigma_H$ the recovery rates of the exposed,
infected and hospitalized, respectively. 

Finally, the model takes into account the changes in the contact number, which also includes the
virulence (infectivity) and other effects related to changes in contact,
of the coronavirus via a novel differential inclusion for the infection rate coefficient 
$\beta$. These changes may be due to the changes in social behavior, changes in weather (allowing more
people to be outside where the infection rates are smaller), or minor mutations of the virus.
We assume that $\beta$ increases in proportion
to the infected population, with rate coefficient $\delta_{*}\geq 0$, and 
decreases in proportion to the population immunity and the number of those
recovered, with rate constant $\delta^{*}\geq 0$.  
This is expressed in the differential inclusion (\ref{eqn210}) and a detailed explanation 
can be found below. However, this differential inclusion is {\it ad-hoc} and different 
formulations will be investigated in the future when  
 a deeper understanding of the virus dynamics emerges. Moreover, it may be of interest to add in the future
 the effects of vaccination.

With the notation and assumptions given above, our model for the COVID-19 pandemic is defined below.
\begin{model}
\label{model}
Find eight functions $(S, E_{fc}, E_{pc}, I_{fc}, I_{pc}, H, R, \beta)$, defined on $[0,T]$, that satisfy the
following system of ODEs,
\begin{align}
\frac{dS}{dt} & = p_{S}- \Gamma S -\mu S, \label{eqn21} \\
\frac{dE_{fc}}{dt} & =p_{Efc}+ \theta \Gamma S  -\left(\gamma_{fc}+\gamma^{+} 
+\sigma_{E}+\mu - \frac{\gamma_{E}}{N} E_{pc}\right) E_{fc},  \label{eqn22}\\
\frac{dE_{pc}}{dt} & =p_{Epc}+ (1-\theta)  \Gamma S  -\left(\gamma_{pc}+\gamma^{-} 
+\sigma_{E} + \mu + \frac{\gamma_{E} }{N}E_{fc}\right) E_{pc},  \label{eqn23}\\
\frac{dI_{fc}}{dt} & =p_{Ifc}+  \gamma_{fc}E_{fc} +\gamma^{-}E_{pc}
- \left(\delta+\sigma_{I}  +d_{I} +  \mu - \frac{\gamma_{I}}N  I_{pc} 
\right) I_{fc}, \label{eqn24}\\
\frac{dI_{pc}}{dt} & =p_{Ipc}+  \gamma_{pc}E_{pc} +\gamma^{+}E_{fc}
- \left(\delta+\sigma_{I}  +d_{I} +  \mu + \frac{\gamma_{I}}N  I_{fc}\right) I_{pc},  \label{eqn25}\\
\frac{dH}{dt} & =p_{H}+ \delta( I_{fc} +  I_{pc})- 
( \sigma_H +d_H+\mu) H,  \label{eqn26}\\
\frac{dR}{dt} & = p_{R} + \sigma_{E}(E_{fc} + E_{pc})+ \sigma_{I}(I_{fc}
+ I_{pc})+\sigma_H H -\mu R, \label{eqn27}\\
N &=S +E_{fc} +E_{pc} +I_{fc} +I_{pc} +H +R,
\label{eqn28}\\
\Gamma&=\frac{\beta}{N}(\epsilon_{Efc} E_{fc} +\epsilon_{Epc} E_{pc}+
\epsilon_{Ifc}I_{fc} + \epsilon_{Ipc} I_{pc} + \epsilon_H H) ,   \label{eqn29} \\
\frac{d\beta}{dt} & =\delta_{*}\Gamma - \delta^{*}\beta \frac{R}N -\zeta,  
\quad \zeta\in \partial I_{[\beta_*,\beta^*]}(\beta), \label{eqn210}
\end{align}
together with the initial conditions,
\[
S(0)= S_0,\,  E_{fc}(0)=E_{fc0},\, E_{pc}(0)=E_{pc0},\,
I_{fc}(0)=I_{fc0},\, I_{pc}(0)=I_{pc0},
\]
\begin{equation}\label{eqn211}
 H(0)=H_{0},\,R(0)=R_0,
\end{equation}
\[
\beta(0)=\beta_{0}\in [\beta_{*}, \beta^{*}].
\]
\end{model}
Here, $S_0>0$ is the initial population at the breakout of the pandemic; 
$E_{fc0}, E_{pc0}$, $I_{fc0}, I_{pc0}, H_{0}$ and $R_0$  are nonnegative initial subpopulations, 
and $\beta_{0}$ is the initial transmission rate which can be estimated from the data.

In practice, one typically assumes that $S_0=N(0)>0$, so that at the
start of the pandemic there are only susceptibles, and all the other populations vanish.
However, for the sake of generality, we allow initially the other subpopulations to be nonnegative.
 Specifically, if the starting point is later than
the very first day of the pandemic, 
\[
N(0)=S_{0}+E_{fc0}+ E_{pc0}+I_{fc0}+ I_{pc0}+ H_{0}+R_0,
\]
so that (\ref{eqn28}) holds initially.
A summary of the definitions of the model parameters is given in Table\,\ref{tab:1a}.
\medskip

Equation (\ref{eqn21}) describes the rate of change, per day,  of the susceptible population. 
The second term on the right-hand side is the rate at which the susceptibles become infected 
by contact with  exposed, infected, and hospitalized individuals. We emphasize that  
the exposed populations include the asymptomatics that  may cause infection.
 The last term describes the `natural' (unrelated to the pandemic) mortality. 

The rest of the equations, except  (\ref{eqn210}), have a similar structure and interpretation.  
For instance, in (\ref{eqn26}) the rate of change, per day, of the
hospitalized is the sum of those who arrive from outside (say overflow in other locations),
$p_{H}$, and   fully compliant and partially compliant 
infectives whose illness becomes severe and need hospitalization, 
$\delta(I_{fc} + I_{pc})$, minus those who recovered 
on that day, $\sigma_{H}H$, those who died naturally,  $\mu H$, and those who died 
because of the pandemic, $d_{H} H$.  

We note that in our model, the usual expression for 
the probability that one susceptible is infected (per day), $\beta I/N$,
is replaced with $\Gamma$. This means that the values of $\beta$ must be restricted,
which we describe next.

Next, we describe equation~(\ref{eqn210}), actually a differential inclusion, for the change in the
infectiveness of the SARS-CoV-2 virus. We assume  a `simple' linear 
relationship between the infection rate coefficient $\beta$ and the fractions of those
who carry the virus and the recovered. In particular, an increase in the fraction who have the virus  and can infect others, $(E_{fc}+E_{pc}+I_{fc}+I_{pc}+H)/N$, is assumed to
make the virus more infective, while an increase in the fraction of recovered, $R/N$, is assumed to
increase the immunity of the population leading to a decrease in the disease virulence. 
Finally, we use a subdifferential (explained below) to guarantee that $\beta$ remains in the admissible set $[\beta_{*},\beta^*]$, where $\beta_*$ and $\beta^*$ are appropriate values such that $0\leq \beta_*<\beta^*$.
Let $I_{[\beta_*,\beta^*]}(\beta)$ be the indicator function of the interval
$[\beta_*,\beta^*]$, which vanishes when $\beta$ is in the interval and has the value
$+\infty$, otherwise.
Then, its subdifferential, $\partial I_{[\beta_*,\beta^*]}(\beta)$, is the set-valued function or multifunction 
\begin{equation}
\label{subdiff}
\partial I_{[\beta_*,\beta^*]}(\beta)=
 \begin{cases}
(-\infty, 0]      & \text{ if } \beta=\beta_*, \\
   0   & \text{ if } \beta_{*}< \beta < \beta^{*}, \\
  [0, \infty)      & \text{ if } \beta=\beta_*, \\
   \emptyset   & \text{otherwise}.
\end{cases}
\end{equation} 
Adding this term guarantees that
$\beta_{*}\leq \beta(t)\leq \beta^{*}$ for $0\leq t\leq T$. Indeed, when $\beta \in (\beta_*,\beta^*)$
then $\zeta=0$ and (\ref{eqn29}) is just a rate equation. When $\beta(t)=\beta_{*}$ 
then there exists an element $-\zeta \in [0, \infty)$ that prevents $\beta$ from 
becoming smaller than $\beta_{*}$; and when $\beta(t)=\beta^{*}$ 
there exists an element $-\zeta \in (-\infty, 0]$ that prevents $\beta$ from exceeding
$\beta^{*}$.

The number of active `documented' or identified cases $A(t)$, at time $t$,  is
\begin{equation*}
\label{eqn212}
A(t)=I_{fc}(t)+I_{pc}(t)+H(t),
\end{equation*}
and, following \citep{Worldometer}, we denote by $A_{m}(t)=I_{fc}(t)+I_{pc}(t)$ 
and $A_{c}(t)=H(t)$ those 
with an active mild condition, documented but not serious, which can be with no symptoms, mild or medium symptoms,
and those with serious or critical conditions, 
respectively.

Next, $CC(t)$, the cumulative or total number  of cases, up to time $t$, is
\begin{equation}
\label{eqn215}
CC(t) = A(t) + CR(t) + CD(t),
\end{equation}
where the cumulative number of the recovered $CR$ is given by
\begin{align*}
\label{eqn216}
CR(t) =R_0+ \int_0^t & \Big(\sigma_{E}\big(E_{fc}(\tau) + E_{pc}(\tau)\big)\\ 
& + 
 \sigma_{I}\big(I_{fc}(\tau)  + I_{pc}(\tau)\big)
+ \sigma_{H}H(\tau)\Big) d\tau;
\end{align*} 
the daily number of deaths caused by the disease is
\begin{equation*}
\label{eqn217}
D(t)= d_{I}( I_{fc}(t) +   I_{pc}(t))+ d_H H(t);
\end{equation*}
and the cumulative or total number of deaths $CD(t)$ caused by the disease  is 
\begin{equation}
\label{eqn218}
CD(t)  = D_0 + \int_{0}^{t}D(\tau)\,d\tau.
\end{equation}
The number of new cases on day $t$ is given by 
\begin{equation}
\label{eq:Cday}
C_{day}(t)=CC(t)-CC(t-1).
\end{equation}

Finally, the cumulative number $CA(t)$ of those who recovered from the asymptomatic population is given by
\begin{equation}
\label{eq:CA}
CA(t)= \sigma_{E} \int_0^t \left(E_{fc}(\tau) + E_{pc}(\tau) \right) d\tau,
\end{equation}
assuming that the initial number of the asymptomatic who recovered is zero.

 In this way, the model provides insight into an important aspect of the pandemic that is 
 very difficult to obtain in the field.

\section{Existence and uniqueness of the solution}
\label{sec-exist}
Proving the existence of  a unique solution to the model for each finite time interval is 
a mathematically important next step. Without the
differential inclusion (\ref{eqn210}) and when $\theta$ is a known
continuous and bounded function, the local existence in time follows from the fact that
the functions on the right-hand sides of the equations are locally Lipschitz. Then, 
the existence of the global solution is established by demonstrating that the 
solution stays bounded on every finite time interval.  This includes showing that given 
nonnegative initial conditions, the solution components stay nonnegative, which is not 
only necessary mathematically but also important for the model to be biologically 
relevant.  However, this is not the case in our model. Specifically, we allow $\theta$ 
to be bounded and piecewise continuous, hence not continuous, and $\beta$ is not given
but a solution of a differential inclusion. Therefore, the standard way 
outlined above no longer works, requiring a more sophisticated approach that 
we detail in what follows.

We assume that the input functions $p_{S}(t),\dots, p_{R}(t)$ are bounded, nonnegative, and
smooth and all the parameters, except for $\gamma_E$ and $\gamma_I$, are positive constants.

First, we assume that
$\theta(t)$ is given and $0\leq \theta(t)\leq 1$ for all $t\in [0,T]$, and $\beta$ is a given smooth 
function with values in $[\beta_*, \beta^*]\subset (0,1]$. 
Below, we relax the assumption that $\beta$ is given.

\paragraph{Positivity of $E_{fc}; E_{pc}; I_{fc}; I_{pc}; H; R$. } 
We first establish the non-negativity of the solutions of the system when
the initial data is non-negative (except for $S_0>0$). To that end, let $S_0>0$ and
assume that the other initial conditions  satisfy
$E_{fc0}, E_{pc0}, I_{fc0}, I_{pc0}, H_0, R_0\geq \varepsilon>0$,
and also $p_R  \geq \varepsilon$. Below, we let $\varepsilon\to 0$.
Thus,  there exists $T_1>0$, which we may assume to be the largest time,
such that the local solution exists and is unique on $[0, T_1)$. The 
continuity of the solution implies that there is a maximal time $t_*$, satisfying 
$0<t_*\leq T_1$, which may depend on $\varepsilon$, such that  the solution is 
component-wise positive on $[0, t_*)$.
Moreover,  we have that $N>0$ on $[0, t_*]$, since otherwise if $N(t_*)=0$ then
by (\ref{eqn28}), 
\[
S(t_*)=E_{fc}(t_*)=E_{pc}(t_*)=I_{fc}(t_*)=I_{pc}(t_*)=H(t_*)=R(t_*)=0,
\]
and it follows from (\ref{eqn27}) that at $t=t_*$,
\[
\frac{dR}{dt}  \geq  \varepsilon>0,
\]
which means that $R$ is an increasing function at $t_*$, and since  $R(t)>0$
on $[0, t_*)$, it is positive on $[0, t_*]$, and then $N>0$ on $[0, t_*]$.

Then (\ref{eqn29}) implies that
\begin{equation}
\label{eqn217a}
0\leq \Gamma  < 5 \beta ^{\ast } \max \left( \epsilon
_{E_{fc}},\epsilon _{E_{pc}},\epsilon _{I_{fc}},\epsilon _{I_{pc}},\epsilon
_{H}\right):= \Gamma^*, 
\end{equation}
since $E_{fc}, E_{pc},I_{fc}, I_{pc}, H<N$ and $\beta\leq \beta^{*}$  on $[0, t_*]$.
It follows from (\ref{eqn21}) that
\[
\dfrac{dS}{dt}\geq p_S-(\Gamma^*+\mu)S,
\]
on $[0, t_*]$. Since $S_0>0$ and $p_S\geq0$, then on $[0, t_*]$,
\[
S(t)\geq S_0 \exp\left(-(\Gamma^*+\mu)t\right)\geq S_0\exp(-(\Gamma^*+\mu)t_*)>0.
\]

Next, we introduce the following notation for the sake of simplicity,
\begin{align}
\widehat{\gamma}_{fc} &:= \gamma_{fc}+\gamma^{+} +\sigma_{E}+\mu, \label{eq:hat1}\\
\widehat{\gamma}_{pc} &:=\gamma_{pc}+\gamma^{-} +\sigma_{E} + \mu,\label{eq:hat2}\\
\widehat{\delta}&:= \delta+\sigma_{I}  +d_{I} +  \mu,\label{eq:hat3}
\end{align}
and assume that the parameters satisfy
\begin{equation}\label{eq:assumption}
    |\gamma _{E}|<\min\{\widehat{\gamma}_{fc}, \widehat{\gamma}_{pc}\}, \quad |\gamma _{I}| <\widehat{\delta},
\end{equation}
which is necessary to prove boundedness but also makes the model biologically relevant.

Next, we consider (\ref{eqn22}) and since $\theta \Gamma S\geq 0$ and $E_{pc}< N$, we obtain  
\[
\frac{dE_{fc}}{dt} \geq p_{Efc}  -(\widehat{\gamma}_{fc}+ |\gamma_E|)E_{fc},
\]
and since $p_{Efc}\geq 0$, we find that  $E_{fc}(t)\geq E_{fc0}\exp(-(\widehat{\gamma}_{fc} +|\gamma_E|)t_*)>0$
on $[0, t_*]$.
We now consider (\ref{eqn23}), and since $p_{Epc}\geq 0$, $(1-\theta)\Gamma S\geq 0$ 
and $E_{fc}<N$, we have
\[
\frac{dE_{pc}}{dt} \geq p_{Epc}  -\left(\widehat{\gamma}_{pc} + |\gamma_{E}| \right) E_{pc},
\]
therefore, $E_{pc}(t)\geq E_{pc0}\exp\left(-(\widehat{\gamma}_{pc} 
+ |\gamma_{E}|)t_* \right)>0$ on $[0, t_*]$.

The arguments concerning (\ref{eqn24}) and (\ref{eqn25}) are similar.  Since
$p_{Ifc}\geq 0$, by omitting all the nonnegative terms in (\ref{eqn24}), we find,
\[
\frac{dI_{fc}}{dt} \geq p_{Ifc}  -(\widehat{\delta}+|\gamma_I|) I_{fc},
\]
and hence, $I_{fc}(t)\geq I_{fc0}\exp(-(\widehat{\delta}+|\gamma_I|)t_*)>0$ on $[0, t_*]$.
Similarly, it follows from (\ref{eqn25}) that
$I_{pc}(t)\geq I_{pc0}\exp(-(\widehat{\delta}_{fc}+|\gamma_I|)t_*)>0$ on $[0, t_*]$.

We turn to (\ref{eqn26}). Since $p_{H}\geq 0$ and $\delta_{fc} I_{fc} 
+ \delta_{pc} I_{pc}>0$, we have
\[
\dfrac{dH}{dt}\geq  -  ( \sigma_H +d_H+\mu) H,
\]
and so $H(t)\geq H_0\exp\left(-(\sigma_H +d_H+\mu)t_*\right)>0$ on $[0, t_*]$. 
Finally, it follows from (\ref{eqn27}) that 
\[
\dfrac{dR}{dt}> -\mu R,
\]
and hence $R(t)>R_0\exp(-\mu t_*)>0$ on $[0, t_*]$.

We note that all the decay rates in the estimates above depend on the problem data but
are independent of $\varepsilon$, then all the variables 
$S(t),\dots , R(t)$, as well as $N(t)$ and $\Gamma(t)$, are
positive on the closed interval $[0, t_*]$, and since the solution is continuous, we conclude that 
they are strictly positive on $[0, T_*]$
for some  $t_*< T_*<T_1$,  and by extension $T_*=T_1$.  
Indeed, if $S(t_*)=\alpha_S>0$ and $S$ is defined and continuous at $t_*$, then there exists an interval
$(t_*-\xi, t_*+\xi)$, for some $\xi>0$, such that $S(t)\geq \alpha_S/2$ on $(t_*-\xi, t_*+\xi)$.

We conclude that when $\varepsilon>0$, i.e., the initial
conditions are positive, then the solution is positive as long as it exists.
However, it is noted that $T_*$ does not depend on $\varepsilon$, and therefore, we have the following summary.
 \begin{proposition}
 \label{prop1}
 Assume that \eqref{eq:assumption} holds, $S_0>0$, the other initial conditions are nonnegative, 
 and $\theta$ and $\beta$ are bounded and positive. Then, the solution of system (\ref{eqn21})-(\ref{eqn27}) 
 is positive as long as it exists.
 \end{proposition}
Now that it is established that under the assumptions of Proposition \ref{prop1}, the solution is positive as long as it exists, the only way it can cease to exist is when one or more of the variables approaches infinity in finite time.  

\paragraph{Boundedness of $S, E_{fc}, E_{pc}, I_{fc}, I_{pc}, H, R$ on finite intervals}
In the second step, under the assumptions of Proposition \ref{prop1}, we show that each of 
$S(t),\dots, R(t)$ is bounded on every finite time interval, and therefore, the solution cannot approach infinity in finite time. 
This is sufficient to show that the solution exists on each finite time interval. We begin with the interval
of existence, $[0, T_1)$.

We note that since $p_{S}(t)$ is bounded and $\Gamma>0$, equation 
(\ref{eqn21}) shows that $S(t)$ is bounded on every finite time interval it exists on.
Using the definition of $\widehat{\gamma}_{1}$ given by \eqref{eq:hat1} in equation 
(\ref{eqn22}), $0\leq \theta \leq 1$, the bound on $\Gamma$ given by \eqref{eqn217a},
and $0< {E_{pc}}/{N}<1$,  yield
\begin{equation*}
\label{eqn32a}
\frac{dE_{fc}}{dt}   <p_{E_{fc}}+\Gamma^{*} S- (\widehat{\gamma}_{fc}-|\gamma _{E}|)E_{fc}.
\end{equation*}
Then, by assumption~\eqref{eq:assumption} that  $|\gamma_E|<\widehat{\gamma}_{fc}$ and since $p_{E_{fc}}$ 
and $S$ are bounded, it follows that $E_{fc}$ is bounded. 
We now consider (\ref{eqn23}). Using the definition of $\widehat{\gamma}_{pc}$ given by 
\eqref{eq:hat2},  $0\leq 1-\theta\leq 1$, and \eqref{eqn217a},
\begin{equation*}
\label{eqn33a}
\frac{dE_{pc}}{dt} < p_{E_{pc}}+\Gamma^{*} S- (\widehat{\gamma}_{pc}-|\gamma_E|) E_{pc}. 
\end{equation*}
Since $p_{E_{pc}}$ and $S$ are bounded and  $\widehat{\gamma}_{pc}>|\gamma_E|$, it follows that $E_{pc}$ is bounded.  
Next, using the definition of $\widehat{\delta}$ from \eqref{eq:hat3} and $0<I_{pc}/N <1$ in (\ref{eqn24}) and (\ref{eqn25}), we obtain 
\begin{align*}
\frac{dI_{fc}}{dt}& <p_{I_{fc}}+\gamma _{fc}E_{fc}+\gamma ^{-}E_{pc}-(\widehat{\delta} - |\gamma _{I}| )I_{fc}, \\
\frac{dI_{pc}}{dt}& <p_{I_{pc}}+\gamma _{pc}E_{pc}+\gamma ^{+}E_{fc}- (\widehat{\delta}-|\gamma_I|) I_{pc}.  
\end{align*}
Since $|\gamma _{I}|<\widehat{\delta}$,  \eqref{eq:assumption}, 
and  $p_{I_{fc}}, p_{I_{pc}}, E_{fc}$ and $E_{pc} $ are all bounded, the
boundedness of $I_{fc}$ and $I_{pc}$ follows.
The boundedness of $H$ and $R$ follows directly from equations 
\eqref{eqn26}, \eqref{eqn27} and the fact that $p_H, p_R$ and $E_{fc}, E_{pc}, I_{fc}$, $I_{pc}$ 
are all bounded independently of the choice of $\beta\in \left[ \beta _{\ast },\beta^{\ast }\right]$.

We summarize the result as follows.
 \begin{proposition}
 \label{prop2}
 Assume that $S_0>0$, the other initial conditions are nonegative, and $\theta$ and $\beta$ are 
 bounded and positive. Moreover, assume that
 \[
 |\gamma_E|<\min\{\widehat{\gamma}_{fc}, \widehat{\gamma}_{pc}\},\qquad |\gamma_I|<\widehat{\delta}.
 \]
  Then, the solution of system (\ref{eqn21})-(\ref{eqn27}) is positive and bounded on every finite interval.
 \end{proposition}

Before proceeding with the existence proof, we note that 
$\Gamma(t)$ is bounded and positive and the boundedness results above 
and equation (\ref{eqn28}) imply that $N(t)$ is bounded and positive on every finite 
interval. Furthermore, the differential set-inclusion  \eqref{eqn210}  can be rewritten as
\begin{equation}
\frac{d\beta }{dt}+\delta ^{\ast }\frac{R}{N}\beta +\partial \phi \left(
\beta \right) \ni \delta _{\ast }\Gamma ,\qquad 
 \beta \left( 0\right) =\beta
_{0}\in \left[ \beta _{\ast },\beta ^{\ast }\right] ,   \label{eqn38}
\end{equation}
where
 $\phi(\beta)=I_{[\beta_*,\beta^*]}(\beta)$,
and $\partial \phi(\beta)$ is its subdifferential, \eqref{subdiff}.
Since $\phi$ in \eqref{eqn38} is convex and lower semicontinuous, its 
subdifferential $\partial \phi$ forces $\beta $ to remain in 
the interval $\left[ \beta _{\ast },\beta^{\ast }\right]$.

\vskip4pt
\paragraph{Existence and uniqueness}
To use the powerful tools of convex analysis for the existence and uniqueness proof, 
we first reformulate the problem in an abstract form. Letting
\[
\bx=(S, E_{fc}, E_{pc}, I_{fc}, I_{pc}, H, R),
\]
the system defined by equations  (\ref{eqn21})--(\ref{eqn27}) together with (\ref{eqn28}) and 
\eqref{eqn29}, can be written  as
\begin{equation}
\label{eqn39}
\bx ^{\prime }=\bbF\left( t,\bx, \beta, \theta \right),\quad  \bx\left( 0\right) =\bx _{0},
\end{equation}
where $\bbF$ is Lipschitz in  $\bx, \beta $ and $\theta$; $\beta $  is a given continuous function having
values in $\left[ \beta _{\ast },\beta ^{\ast }\right] $;
and $\theta$ and is bounded with values in $[0,1]$.

We first assume that $\theta$ is continuous. Below, we relax this 
assumption and allow it to be piecewise continuous. Let $\beta$ and $\hat{\beta}$ be two continuous functions 
such that $\beta(t), \hat{\beta}(t)\in \left[ \beta _{\ast },\beta ^{\ast }\right] $ 
for all $t\in \left[0,T\right] $. We let  $\bx $ and $\hat{\bx}$ be the solutions to \eqref{eqn39}
corresponding to $\beta $ and $\hat{\beta}$, respectively. Then, using straightforward
 computations, as in \citep{Ken17,Bellman95,Hale09}, we obtain
 a constant $C>0$, independent of $\beta$ and $\hat{\beta}$, such that
\begin{align*}
&\left\vert \bx \left( t\right) - \hat{\bx}\left( t\right)
\right\vert ^{2} \leq C\int_{0}^{t} \vert \bbF\left( s,\bx,\beta \right) -
\bbF\left( s,\hat{\bx},\hat{\beta}\right)  \vert ^{2}ds \\
&\leq 2C\int_{0}^{t}\left(  \vert \bbF\left( s,\bx,\beta
\right) -\bbF\left( s,\bx,\hat{\beta}\right)  \vert
^{2}+ \vert \bbF\left( s,\bx,\hat{\beta}\right) -\bbF
\left( s,\hat{\bx},\hat{\beta}\right)  \vert ^{2}\right) ds \\
&\leq 2CK^{2}\int_{0}^{t}\left(  \vert \beta -\hat{\beta} \vert
^{2}+\vert \bx-\hat{\bx}\vert ^{2}\right) ds.
\end{align*} 
Here, $K$ is an appropriate Lipschitz constant that depends on the estimates obtained 
above. Then, an application of Gr\"{o}nwall's inequality  shows that after modifying the
constants,
\begin{equation}
\left\vert \bx \left( t\right) -\hat{\bx}\left( t\right)
\right\vert ^{2}\leq CK^{2}\int_{0}^{t} \vert \beta -\hat{\beta}
 \vert ^{2}ds.  \label{eqn310}
\end{equation}

The Lipschitz continuity of $\bbF$ ensures the
existence and uniqueness of a solution to the initial value problem with
given $\beta $ and continuous $\theta$ (see e.g., Theorem 2.4 \citep{BD86}).

We summarize the result of the discussion above as follows. 
\begin{proposition}
\label{lem:existence1}
Assume that $\beta $ is a continuous function having values in 
$\left[\beta_{\ast },\beta ^{\ast }\right]$. Then, there exists a unique solution to the
initial value problem defined by (\ref{eqn21})--\eqref{eqn27}, together with \eqref{eqn28}
and \eqref{eqn29}, and the initial conditions \eqref{eqn211}.
Moreover, if $\bx,\hat{\bx}$ are two solutions corresponding to $\beta $ and $\hat{\beta},$ 
then  \eqref{eqn310} holds for a constant $C$ that is independent of $\beta $.
\end{proposition}
Now, let $\beta \in C\left( \left[ 0,T\right] \right) $ such that $\beta \left(
t\right) \in \left[ \beta _{\ast },\beta ^{\ast }\right] $ for all $t$. 
Then, we use the solution that exists by Proposition~\ref{lem:existence1} to 
construct a solution to the evolution inclusion \eqref{eqn38}, based on the 
following well-known theorem of 
Br\'{e}zis \citep{brezis}. In our case the Hilbert space is $\mathbb{R}$ and we assume
that $\beta _{0}\in[\beta_{*}, \beta^{*}]$.
\begin{theorem}
\label{thm3}Let $H$ be a Hilbert space. Let $f\in L^{2}\left( 0,T;H\right) .$
Let $\phi $ be a lower semicontinuous convex proper function defined on $H$
and  $\beta_0$ be in the domain of $\phi$. Then, there exists a unique
solution $\beta \in L^{2}\left( 0,T;H\right) ,\beta ^{\prime }\in
L^{2}\left( 0,T;H\right) ,$ to
\begin{equation*}
\beta ^{\prime }\left( t\right) +\partial \phi \left( \beta \left( t\right)
\right) \ni f\left( t\right)\quad  \text{  a.e. } t,\qquad  \beta \left( 0\right)
=\beta _{0}.
\end{equation*}
\end{theorem}
We also have the following result:
\begin{corollary}
\label{corollary1} In addition to the assumptions of Theorem~\ref{thm3}, suppose $f\left( t\right) $ is
replaced with $f\left( t\right) \beta $, where $f\in L^{\infty }\left(
0,T;H\right) $. Then, there exists a unique solution to the resulting
inclusion.
\end{corollary}
\begin{proof}
Let $\hat{\beta}\in C\left( \left[ 0,T\right] ;H\right)$ and let
$F(\hat{\beta}) $ be the solution of 
\begin{equation}
\label{eqn32fbeta}
\beta ^{\prime }\left(
t\right) +\partial \phi \left( \beta \left( t\right) \right) \ni f\left(
t\right) \hat{\beta}.
\end{equation}
Then,  standard manipulations and the monotonicity of the subgradient,
for $F(\hat{\beta}) $ and $F(\bar{\beta}) $,   yield
\begin{equation*}
\frac{1}{2} \vert ( F ( \hat{\beta}) -F( \bar{\beta})) ( t)  
\vert _{H}^{2}\leq \int_{0}^{t}\left(
f( s)( \hat{\beta}( s) -\bar{\beta}(s)) ,( F( \hat{\beta}) 
-F( \bar{\beta})) ( s) \right) ds.
\end{equation*}
This implies,
\begin{equation*}
\vert ( F( \hat{\beta}) -F( \bar{\beta})) ( t) \vert _{H}^{2}\leq C\int_{0}^{t}\vert
\hat{\beta}( s) -\bar{\beta}( s) \vert ^{2}ds
\end{equation*}
and shows that a sufficiently high power of $F$ is a contraction map.
Hence, $F$  has a unique fixed point that is the unique solution of
the evolution inclusion (\ref{eqn32fbeta}). 
\end{proof}

Now, we turn to the whole problem defined by (\ref{eqn21})--(\ref{eqn211}) 
and construct a mapping $\Theta :C([0,T]) 
\rightarrow C([ 0,T]) $ as follows. Let $\bar{\beta} \in C\left( \left[ 0,T
\right] \right) $ having values in $\left[ \beta _{\ast },\beta ^{\ast }
\right] $. Then, it follows from Proposition~\ref{lem:existence1} that
the solution to the initial value problem for such 
fixed $\bar{\beta}$ is unique. Now, we define a map 
$\Theta :C([ 0,T]) \rightarrow C([ 0,T]) $
where $\Theta ( \bar{\beta} ) =\beta ,$ is the solution of
(\ref{eqn38}) for the given $\bar{\beta}$. 
We write the equation for $\beta$, given $\bar{\beta}$, as
\begin{equation*}
\frac{d\beta }{dt}+\delta ^{\ast }\frac{R( \bar{\beta}) }{N(
\bar{\beta}) }\beta +\partial \phi ( \beta) \ni \delta
_{\ast }\beta F( \bar{\beta}) ,\qquad \beta( 0) =\beta
_{0}\in [ \beta _{\ast },\beta ^{\ast }] .
\end{equation*}
Since the estimates above do not depend on the choice of $\bar{\beta}$,
as long as it has values in $\left[ \beta _{\ast },\beta ^{\ast }\right] ,$ the 
differential inclusion is of the form
\begin{equation*}
\frac{d\beta }{dt}+\partial \phi ( \beta ) \ni G( t, \bar{\bx},
\bar{\beta}) ,\qquad  \beta \left( 0\right) =\beta _{0}\in \left[
\beta _{\ast },\beta ^{\ast }\right],
\end{equation*}
where $ \bar{\bx}$ is the solution to the initial value problem with given $\bar{\beta}$, and
$G$ is a Lipschitz continuous function in both $\bar{\bx}$ 
and $\bar{\beta}$. 

To proceed, we let  $\bar{\beta}_1$ and $\bar{\beta}_2$ be given
with the properties as above, and let $\Theta ( \bar{\beta}_1) \equiv
\beta _{1}$ and $\Theta (\bar{\beta}_2) \equiv \beta _{2}$. Then,
it follows from the inclusion and the monotonicity of $\partial \phi$ and routine
computations that there exists a constant $C>0$, 
independent of $\beta$ and  a suitable Lipschitz constant $K$, such that
\begin{equation*}
\frac{1}{2}\vert \beta _{1}( t) -\beta _{2}( t)
\vert ^{2}\leq 2CK^{2}\int_{0}^{t}(\vert \bar{\beta}_1-\bar{
\beta}_2 \vert ^{2}+\vert \bar{\bx}_1-\bar{\bx}_2\vert^{2}) ds.
\end{equation*}
It follows from (\ref{eqn310}), after modifying the constants,  that
\begin{multline*}
\vert \Theta ( \bar{\beta}_1)( t) -\Theta (
\bar{\beta}_2)( t)\vert ^{2}\equiv \vert \beta
_{1}( t) -\beta _{2}( t)\vert ^{2}
\\
\leq C\left(\int_{0}^{t}(\vert \bar{\beta}_1-\bar{\beta}_2 \vert
^{2}\,ds +\int_0^t\int_{0}^{s} \vert \bar{\beta}_1(\tau)-\bar{\beta}_2
(\tau)\vert ^{2}\,d\tau ds \right).
\end{multline*}
Therefore,  a sufficiently high power of $\Theta $ is a contraction mapping, and it has 
a unique fixed point that is the solution to the problem. This completes the
proof of the existence and uniqueness of the solution of the full problem
assuming the given function $\theta $ is continuous. 

To take into account a piecewise continuous $\theta$, we assume that 
there are finitely many non-overlapping time intervals $\left\{ I_{i}\right\}
$ and continuous functions $\theta _{i}\left( t\right) $ defined on $
\overline{I_{i}}$ such that $\theta \left( t\right) =\theta _{i}\left(
t\right) $ on the interior of $I_{i}$. Then, on each interval where $\theta$ is continuous,
we have a unique solution obtained as above, and it is straightforward to piece
these into a global solution of the problem.

This leads to the main mathematical result in this work. 
\begin{theorem} Assume that $\theta$ is bounded and piecewise continuous and
\begin{equation}
\label{thm}
|\gamma_{E}|< \min\{\widehat{\gamma}_{fc}, \widehat{\gamma}_{pc}\},\qquad
|\gamma_{I}| < \widehat{\delta},
\end{equation}
where $\widehat{\gamma}_{fc}$, $\widehat{\gamma}_{pc}$,  $\widehat{\gamma}_{fc}$ are 
defined as in \eqref{eq:hat1}--\eqref{eq:hat3}.
 Then, there exists a unique solution to problem \eqref{eqn21}--\eqref{eqn211} on 
 every finite time interval $[0, T]$.
\end{theorem}
We note that in the model the assumptions on $\gamma_{E}$ and $\gamma_{I}$ make sense,
although it is possible to remove them at the expense of considerable additional effort 
in obtaining the estimates.

Finally, since $\beta$ has values in $\left[ \beta _{\ast },\beta
^{\ast }\right] $ and its derivative is in $L^{2}\left( 0,T\right) $, it follows that  
$\beta $ is H\"older continuous with exponent $1/2$. Therefore, the solution is
at least  in $C^{1,1/2}(0,T)$.  It seems plausible that the regularity is
higher since $\beta'\in L^{\infty}\left( 0,T\right)$ and then
it is H\"older continuous with exponent $1$, however, we leave  the 
question open as part of further study of the solution's regularity.

\section{Randomness in system parameters}
\label{sec-random}
For the sake of completeness, this short section provides a rather abstract discussion 
about adding randomness to the system parameters. This allows a better understanding of 
the model's dependence on the parameter values.  Moreover, to use the model as a 
predictive tool, it is crucial to find out how parameter changes affect model predictions. 
Small changes in the solution that are caused by small changes in a parameter indicate that 
there is low sensitivity to the parameter and an approximate value is sufficient for 
acceptable predictions,  while considerable changes in the solution caused by small changes 
in the parameter values indicate that a more precise parameter value is needed to 
obtain reliable predictions. It may also indicate that the model is unstable or the 
process itself is unstable, in which case attempts at prediction may be of little
use.

We now introduce randomness into the system parameters. For the sake of generality, 
we note that we have 38 system parameters including the 7 initial conditions, as listed in Table~\ref{tab:1a} and let  
the {\it probability space} be  
$(\Omega, \mathcal{F}, P)$, where $\Omega$ is the {\it sample space}, a box
centered at the origin of $\mathbb{R}^{38}$;   
$\mathcal{F}$ is the Borel $\sigma$-algebra, and $P$ is a general probability function. 
We let 
$\overline{\omega} \in \Omega$ be the random variable  and define 
\[
\omega = \widehat{\omega}+\overline{\omega},\qquad \left(\omega \in \widehat{\omega}+\Omega\right),
\]
where $\widehat{\omega}$ denotes the vector containing the optimized parameters.
We note that the choice of $\Omega$ is such that $\omega\geq 0$ (component-wise).

Next, we consider the modified system \eqref{eqn21}--\eqref{eqn29} and \eqref{eqn211}, in the form
\begin{equation*}
\label{eqn41}
\bx ^{\prime }=\bbF\left( t,\bx, \beta, \omega \right),\quad  \bx _{0}\left( 0\right) =\widehat{\bx}(\omega),
\end{equation*}
 where $\widehat{\bx}(\omega)=\bx _{0}+\bar{\omega},\; \bar{\omega}\in \Omega$,
together with the differential inclusion
\begin{equation*}
\frac{d\beta }{dt}+\widetilde{\delta} ^{\ast }\frac{R(\omega)}{N(\omega)}\beta +\partial \phi \left(
\beta \right) \ni \delta _{\ast }\Gamma(\omega) ,\qquad 
 \beta \left( 0\right) =\beta
_{0}\in \left[ \beta _{\ast },\beta ^{\ast }\right] .   \label{eqn42}
\end{equation*}
Here, $\partial\phi(\beta)=I_{[\beta_*,\beta^*]}(\beta)$ is defined by (\ref{subdiff}).
The function $\bbF$ depends on $\omega $ such that  $\omega \rightarrow
\bbF\left( t,\bx,\beta ,\omega \right) $ is measurable, $\bbF$ is Lipschitz continuous in $\bx$ and $\beta $ and is
continuous in all of the first three variables. 

It follows from the recent results in \citep{KLS16} that the uniqueness of
the solutions for each fixed $\omega$ implies that the functions 
$\omega \rightarrow \bx\left(
\cdot ,\omega \right) ,\omega \rightarrow \beta \left( \cdot ,\omega \right)
$ and $\omega \rightarrow \bbF\left( \cdot ,\bx\left( \cdot
,\omega \right) ,\beta \left( \cdot ,\omega \right) ,\omega \right) $ are
each measurable into $C\left( \left[ 0,T\right] \right) $. Moreover,
 $\omega \rightarrow \beta ^{\prime
}\left( \cdot ,\omega \right) $ is measurable into $L^{2}\left( 0,T\right) $, and then
it follows from  Theorem 3.3 in \citep{KLS16} that $\left( t,\omega \right)
\rightarrow \beta ^{\prime }\left( t,\omega \right) $ can be considered as a
product measurable function.

\section{Stability of the DFE and the EE}
\label{sec-stability}
This section discusses the stability of the critical points of the system:
the disease-free equilibrium (DFE) and the endemic 
 equilibrium (EE) (when it exists), based on the system's Jacobian matrix.
The usual approach, see e.g., \citep{hethcote00, allen07, thieme03}, is to 
derive the {\it basic stability number} $\mathcal{R}_C$. In simple SEIR models, 
$\mathcal{R}_C$ coincides with the basic reproduction number $R_0$, however,  in more complex
models the expression is naturally as complex as the model. When 
$\mathcal{R}_C<1$,  all the eigenvalues of the Jacobian, evaluated at the
DFE, have negative real parts, which indicates that the DFE is asymptotically stable
(i.e., stable and attracting), and there is no EE.
On the other hand, $\mathcal{R}_C=1$ is a bifurcation point, 
and when $\mathcal{R}_C>1$  at least one of the real parts is positive and the DFE
loses its stability. Usually, when the DFE becomes unstable the EE  
appears and is stable and attracting.  However, because of the complexity of our model, 
we did not find a closed-form expression for $\mathcal{R}_C$, so, instead, 
we derived the Jacobian of the system (Appendix \ref{appendix}) 
assuming that $\beta$ is constant and evaluated 
it numerically at the DFE.

To proceed, we assume that $\beta=\beta_0\in [\beta_*, \beta^*]$ and the population is 
constant, that is, the number of `natural' COVID-19 unrelated deaths is balanced by $p_S$, 
and so $p_{S}= \mu N$ and $p_{Efc}, \hdots, p_R$ vanish and the disease-related deaths 
that happened on the way to the DFE are not included, that is, $d_I=d_H=0$, since the DFE 
excludes any infection.  We let the solution 
\begin{equation*}
    \label{eqn51}
\bx(t)=(S(t),E_{fc}(t), E_{pc}(t), I_{fc}(t), I_{pc}(t), H(t), R(t)),
\end{equation*}
represent the trajectory of the system in $\mathbb{R}^{7}_{+}$, for $0\leq t \leq T$.
Then, the DFE contains only susceptibles, that is,
\[
DFE=(N, 0,0,0,0,0,0),
\]
and the force of infection vanishes, $\Gamma =0$.

\vspace{0.1cm}
{\small 
\begin{minipage}{\linewidth}
\captionof{table}{Parameters of the model and their description.} \label{tab:1a}
\begin{tabular}{l*{6}{l}r}
Parameter              & Description\\
\hline
\vspace{-0.15cm}
\\
$N(t)$ & total population (at time $t$), (\ref{eqn28})\\
$\mu$     & `natural' death rate coefficient, fixed \\
$\theta(t) ; (1-\theta(t))$ & fractions of fully and partially compliant exposed \\
$p_{S}$ & recruitment rate of susceptibles \\
$p_{E_{fc}}; p_{E_{pc}}$ & recruitment rates of fully and partially compliant exposed\\
$p_{I_{fc}}; p_{I_{pc}}$ & recruitment rates of fully and partially compliant infectives \\
$p_{H}$ & recruitment rate of hospitalized \\
$p_{R}$ & recruitment rate of recovered \\
$\epsilon_{E_{fc}}; \epsilon_{E_{pc}}$  & factors in the transmission rates by exposed   \\
$\epsilon_{I_{fc}};  \epsilon_{I_{pc}}$  &  factors in the transmission rates by infectives   \\
$\epsilon_H $    & a factor in transmission rate by hospitalized   \\
$\gamma_{fc}; \gamma_{pc}$     & rates of developing clinical symptoms in exposed   \\
$\gamma^{-}$ &  crossing rate from $E_{fc}$ to $I_{pc}$ \\
$\gamma^{+}$  &  crossing rate from $E_{pc}$ to $I_{fc}$\\
$\gamma_{E},\; \gamma_{I}$ & rates between fully and partially compliant exposed and infectives   \\
$\delta$  & rate of hospitalization of infectives \\
$d_{I}$    & disease-induced death rates of  infectives \\
$d_H$     & disease-induced death of hospitalized \\
$\sigma_{E}$ & recovery rate of exposed\\
$\sigma_{I}$      & recovery rate of infectives\\
$\sigma_H$     &  recovery rate of hospitalized  \\
$\delta_*$ &  infectiveness rate increase factor with infected \\
$\delta^*$ &  infectiveness rate decrease factor with recovered \\
$\beta_{*}; \beta^{*}$ & lower and upper bounds on the contact rate\\
$\beta_{0}$ & initial contact rate \\
$E_{fc0}$; $E_{pc0}$  & initial values of exposed populations\\
$I_{fc0}$; $I_{pc0}$ & initial values of infected populations\\
$H_0$; $R_0$ & initial values of hospitalized and recovered
\vspace{0.1cm}\\
\hline\\
\end{tabular}
\end{minipage}
}
\vskip8pt
We note that the model contains 29 parameters and seven initial conditions. However, 
in the computations we use three theta values instead of one
so there are $31$ parameters and seven initial conditions, which  are optimized.
The Jacobian  of the system (see Appendix \ref{appendix}) evaluated at the DFE is 
given by
\begin{equation}
\label{eqn54}
 J(DFE)=
\end{equation}
\[
{\small
\left(
\begin{array}{ccccccc}
-\mu & -\beta \epsilon_{Efc} & -\beta \epsilon_{Epc} & -\beta \epsilon_{Ifc} & 
-\beta \epsilon_{Ipc} & -\beta \epsilon_{H} & 0 \\[3pt]
0 & \theta \beta\epsilon_{Efc}- \widehat{\gamma}_{fc}& \theta \beta\epsilon_{Epc} & 
\theta \beta\epsilon_{Ifc} & \theta \beta\epsilon_{Ipc} & \theta \beta\epsilon_{H} & 0 \\[3pt]
0 & (1-\theta)  \beta\epsilon_{Efc} & (1-\theta)  \beta\epsilon_{Epc} -\widehat{\gamma}_{pc}& 
(1-\theta)  \beta\epsilon_{Ifc} &  (1-\theta)  \beta\epsilon_{Ipc} & (1-\theta)  
\beta\epsilon_{H} & 0 \\[3pt]
0 & \gamma_{fc} & \gamma^{-} & -\widehat{\delta} & 0 & 0 & 0 \\[3pt]
0 & \gamma^{+} & \gamma_{pc} & 0 & -\widehat{\delta} & 0 & 0 \\[3pt]
0 & 0 & 0 & \delta & \delta & -\widehat{\sigma_H} & 0 \\[3pt]
0 & \sigma_{E} & \sigma_{E} & \sigma_{I} & \sigma_{I} & \sigma_{H} & -\mu
\end{array}
\right)
}
\]

where
$\widehat{\gamma}_{fc}$, $\widehat{\gamma}_{pc}$, and $\widehat{\delta}$ are defined in 
\eqref{eq:hat1}--\eqref{eq:hat3} and  
\label{eq:hat}
\begin{align}
\widehat{\sigma}_H&=\sigma_H +d_H+\mu.
\end{align}

We remark that a rather technical and involved analysis is needed to 
establish rigorously that when the DFE is stable and attracting, 
\[
\lim_{t\to \infty} \bx(t) =DFE,
\]
for all nonnegative initial conditions, and the EE does not exist.
And when the DFE loses its stability, the EE appears and is stable and attracting.
We leave this  theoretical question open here and remark about it in Section \ref{sec-con}.

\section{Numerical Simulations}
\label{sec-sims}
 This section describes the approach to our  computer simulations of the model. 
We use \texttt{MATLAB}'s \texttt{ode45} ODEs solver together with \texttt{fmincon} constrained 
optimization routine to determine the value of the optimized model parameters 
that minimize the $\ell^1-$deviation of the model predictions for the total number 
of cases and deaths and the seven-day averaged numbers
of daily cases and deaths from the given data. 
To be specific, to find some of the model parameters, we minimize the  objective function:
\begin{align*}
    \Psi(\omega):=&\sum_{n=1}^{N_{data}}|CC(\omega,t_n)-\widehat{CC}_n|+\sum_{n=1}^{N_{data}}|CD(\omega,t_n)-\widehat{CD}_n|\\
    &+\sum_{n=1}^{N_{data}}|DC(\omega,t_n)-\widehat{DC}_n|+\sum_{n=1}^{N_{data}}|DD(\omega,t_n)-\widehat{DD}_n|,
\end{align*}
where $N_{data}$ is the total number of data points, $\omega\in \mathbb{R}^{38}$ is the parameter set, 
$CC(\omega,t_n)$, $CD(\omega,t_n)$, $DC(\omega,t_n)$ and $DD(\omega,t_n)$ are the model predictions using 
$\omega$ for the cumulative cases, cumulative deaths, daily cases and daily deaths, respectively, on day 
$n$, and $\widehat{CC}_n, \widehat{CD}_n, \widehat{DC}_n$ and $\widehat{DD}_n$ denote the actual 
reported cumulative cases, cumulative deaths, 7-day averages of the daily cases, and daily deaths, 
respectively, on day $n$.

Table~\ref{tab:1a} summarizes the meaning of the model parameters in $\omega$. For the model
to make sense, we impose the following constraints:

\begin{equation}\label{eq:constraint}
\widehat{\gamma}_{fc} + |\gamma_E| <1, \quad \widehat{\gamma}_{pc}  +|\gamma_{E}|<1,\quad 
\widehat{\delta} +|\gamma_{I}|<1, \quad \widehat{\sigma}<1,
\end{equation}
where $\widehat{\gamma}_{fc}, \widehat{\gamma}_{pc}, \widehat{\delta}, \widehat{\sigma}$ are defined in \eqref{eq:hat1}--\eqref{eq:hat3}, and \eqref{eq:hat}. 
%
%
The first condition in \eqref{eq:constraint} guarantees that in equation \eqref{eqn22} the number of those that leave
compartment $E_{fc}$, per day,  does not exceed the number of the individuals in that
compartment. The other conditions in \eqref{eq:constraint} have similar interpretations.

\vskip8pt
\hspace{-0.2cm}
{\small
\begin{minipage}{\linewidth}
\centering
\captionof{table}{Optimized model parameter values for South Korea}
 \label{tab:2}
 \hskip-20pt
\begin{tabular}{l|*{6}{c}r}
parameter              &  optimized value	& units 	& interval \\
\hline
\hline
$N(0)$                                   & $51\, 10^{6}$	& indiv  & fixed  \\
$\mu$                                 & $1.7\,10^{-5}$  	& 1/day & fixed  \\
$\theta_1; \theta_2;\theta_3$         & $0.0101, 0.995, 0.352$ & -  & $[0.01,0.25];\;[0.1,0.995];\;[0.35,0.75]$ \\
$p_{S}$                               & $124$ 	& indiv/day  & $[0, 870]$  \\
$p_{E_{fc}}; p_{E_{pc}}$                  & $0; 6.2359$ 	& indiv/day & $[0,870]$; \; $[0,100]$   \\
$p_{I_{fc}}; p_{I_{pc}}$               & $0.1037; 0.0732$ 	& indivl/day  & $[0,10]$; \; $[0,10]$\\
$p_{H}$                               & $0.0$ 	& indiv/day & $[0,5]$ \\
$p_{R}$                               & $4.6508$ 	& indiv/day & $[0,5]$ \\
$\epsilon_{E_{fc}}; \epsilon_{E_{pc}}$    &  0.05; 0.9	& -  & $ [0.05 ,0.6];\; [0.25,0.90]$  \\
$\epsilon_{I_{fc}};  \epsilon_{I_{pc}}$  &  0.3; 0.4	& -   & $ [0.30,0.85];\; [0.4,0.75]$  \\
$\epsilon_H $                         &  0.005	& -  &  $[0.005, 0.2]$  \\
$\gamma_{fc}; \gamma_{pc}$             & 0.085;\; 0.04	& 1/day & $ [0.01,0.085];\; [0.04,0.15]$  \\
$\gamma^{-}; \gamma^{+}$              & 0.001216;\; 0.085	& 1/day & $ [0.0005,0.0035];\; [0.005,0.085]$  \\
$\gamma_{E},\; \gamma_{I}$            & 0.04486\; -0.04261	& 1/day & $[-0.1,0.1];\; [-0.1,0.1]$   \\
$\delta$             & $0.008$	& 1/day & $[0.001,0.008]$  \\
$d_{I}$                         & $0.0005203$	& 1/day & $[0.00025,0.005]$  \\
$d_H$                                 & $0.04909$	 & 1/day   & $[0.01,0.05]$  \\
$\sigma_{E}$           & $0.05067$	& 1/day   &$ [0.01,0.15]$  \\
$\sigma_{I}$           & $0.15$	& 1/day   &$[0.01,0.15]$  \\
$\sigma_H$     & $0.01819$	& 1/day       & $ [0.005,0.1]$  \\
$\delta_*; \delta^*$  & $0.7292;\; 0.2982$  &  1/(indiv)(day)$^2$ & $ [0.005,0.75];\; [0.05,0.95]$ \\
$\beta_{*}; \beta^{*}$ & $0.05; 0.5$  &  1/(indiv)(day)$^2$ & $[0.001,0.5];\;  [0.001,0.6]$ \\
$\beta_{0}$                           &  $0.1133$ & 1/day & $[0.001, 0.6]$\\
$E_{fc0}$; $E_{pc0}$                  & 103.7;\; 2584 & indiv & $[0,500];\;  [0,5000]$ \\
$I_{fc0}$; $I_{pc0}$                   & 1.079;\; 7.29 & indiv & $[0,100];\;  [0,100]$ \\
$H_0$; $R_0$               & 4.147;\; 10 & indiv & $[0,10];\;  [0,10]$ \\
\hline
\end{tabular}
\end{minipage}}

\vskip8pt

\section{COVID-19 in The Republic of Korea}
\label{sim-SK}
This section presents the model simulations of the pandemic's dynamics 
in South Korea, which is chosen because it is one of the first countries 
that went through the pandemic cycle and the complete unbiased data is available on 
the web, in particular in \citep{Worldometer}. As of Aug.  8, 2020, the pandemic there 
is under control, with minor outbreaks.
The values of the optimized baseline parameters and their intervals of feasibility can be found 
in Table~\ref{tab:2}. The total population $N(0)$ and the COVID-19 unrelated death rate $\mu$ 
are fixed constants obtained from South Korea's population data and are
 not part of the optimization process.
We solved the ODEs of Model~\ref{model} numerically using these optimized parameters  
predicting the pandemic's near future.

As noted above, the values in Table~\ref{tab:2} were obtained by an $\ell^1$ optimization 
routine in MATLAB that compares the model predictions for the total number of cases and the 
total number of deaths  as well as the seven-day averaged daily cases and deaths with the data taken from 
\citep{Worldometer} for the time period of 100 days,  from 15 Feb.  until 25 May 2020. 
It is important to emphasize that the model uses past data and, therefore, cannot predict outcomes 
resulting from \underline{major} changes in the virus behavior, new mutations, change in population behavior, 
new  policies and directives, or the environment.
However, the model predictions for the following 75 days are very close to what was observed
in the field, as we show below. Then, by modifying only $\theta$, the model predictions can fit 
very well for more than a year.

\subsection{Baseline simulations}
\label{SKbaseline}
We start with  the baseline simulation, using the 
parameters in Table 2. To properly describe the South Korean government's response to the
pandemic, we use the following values of $\theta$: the lock-down order 
was declared on day 8, and the relaxation of the directive started on day 81. 
To represent these policy changes in the model, we set
\[
\theta =
\begin{cases}
  \theta_1=0.01   & \text{ if }\ t<8, \\
  \theta_2=0.995   & \text{ if }\ 8\leq t< 81, \\
  \theta_3=0.352   & \text{ if }\ 81\leq t\leq   175.
\end{cases}
\]
These values were obtained by the optimization subroutine and in 
 particular, we note that $\theta=0.995$ for days $8-80$ corresponds 
 to very effective control directives and a highly compliant population.

The cumulative number of COVID-19 cases and the cumulative number of deaths are 
depicted in Figure~\ref{fig:variable-theta-CCCD} for the period of 175 days, 
from 15 Feb.  until 8 Aug. 2020.  The red 
filled circles are the field data for the first 100 days of the pandemic 
(15 Feb. until 25 May), \citep{Worldometer}, which were used to optimize the parameters, 
while the green filled squares depict subsequent data (26 May till 8 Aug.).  The agreement between the
green filled squares and the blue line, which is a measure of the model's predictive ability, 
 is good for the cumulative number of cases, although it slightly 
under-predicts the more recent cases, however, the prediction is remarkable for the total 
deaths.  We note the model does not fully capture the initial  exponential growth in the 
number of cases, while  it does so well with the number of deaths. The
median discrepancy in the cumulative cases is $182$ and the median error in the 
cumulative deaths is $4.4$.

\begin{figure}[h!] 
  \hspace{-0.8cm}
  \includegraphics[width=1.10\textwidth]{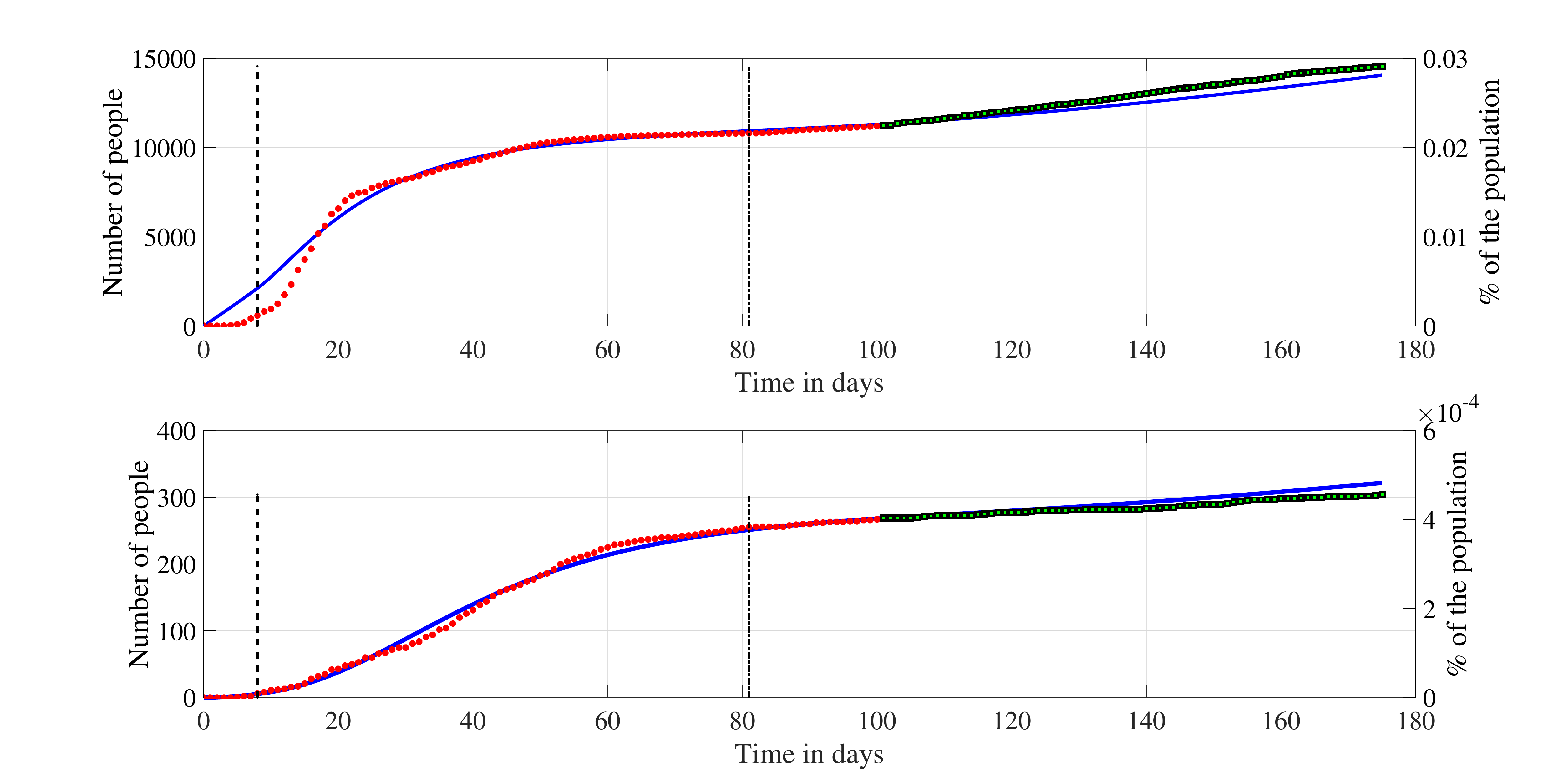}
   \caption{\small Model predictions (blue curves) of the cumulative cases (top) and deaths (bottom); 
   175 days from 15 Feb. until 8 Aug. 2020. The red filled circles are the $100$ days of field 
   data used in the optimization, the green filled squares depict subsequent data. 
   $\theta$ changed from $\theta_1=0.010$ to $\theta_2=0.995$ on day 8 (vertical dashed lines) and to 
   $\theta_3=0.352$ on day 81 (vertical dash-dotted lines). The median 
   error in the cumulative cases was $182$ and the  maximum error was $1847$ (on day 11); the median error 
   in the cumulative deaths was $4.4$ and the maximum error was $18$ (on day 174).}
   \label{fig:variable-theta-CCCD}
\end{figure}

\begin{figure}[h!] 
  \hspace{-1.0cm}
  \includegraphics[width=1.15\textwidth]{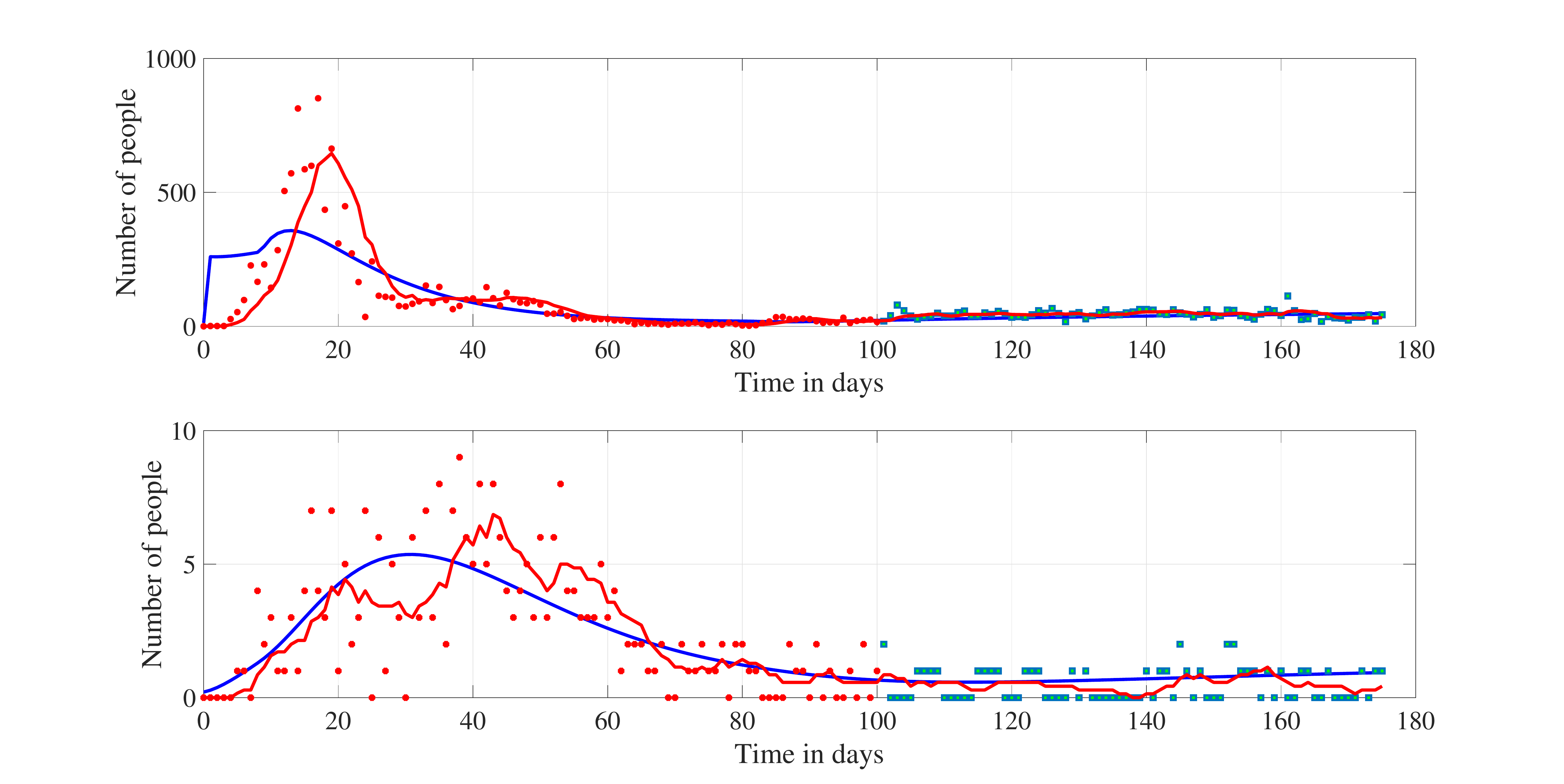} 
   \caption{\small Model predictions (blue curves) of the daily cases of infection (top) and deaths (bottom), 
   and seven-day moving averages of the data (red curves); $175$ days from 15 Feb. until 8 Aug. 2020. The 
   red filled circles are the $100$ days of field data 
   used in the optimization, the green filled squares depict subsequent data. The red curves are the seven-day 
   moving averages of the data.
   $\theta$ changed from $\theta_1=0.01$ to $\theta_2=0.995$ on day $8$ and to $\theta_3=0.352$ on day $81$. }
   \label{fig:variable-theta-DCDD}
\end{figure}
Fig.~\ref{fig:variable-theta-DCDD} depicts the daily cases 
and deaths, respectively, for the same simulations. 
To allow for a better comparison, in addition to the data (red filled circles and green filled squares) we also introduce,
similarly to various web publications, a seven-day moving average (red curve), i.e., averaging the 
data over the previous seven days (or a part of them at the beginning). 
Even though there is considerable 
randomness in the daily data while the simulations don't admit randomness, and the numbers of cases 
are rather small,  the daily case predictions follow  the trends surprisingly well.
We conclude that
the model captures well the known details of the pandemic in South Korea. 

Furthermore, the form of the model allows us to investigate the behavior of the other subpopulations, which are usually 
hard to assess and are not reported separately. This provides further details and considerable  insight into 
the disease dynamics.  Fig.~\ref{fig:variable-theta-IE} 
depicts the model predictions of (for the sake of simplicity) the combined daily numbers of exposed,
$E_{fc}+E_{pc}$, and of infected  $I_{fc}+I_{pc}$, noting that the
corresponding data is not available (in open sources). 

\begin{figure}[ht] 
\hspace{-0.9cm}
   \includegraphics[width=1.15\textwidth]{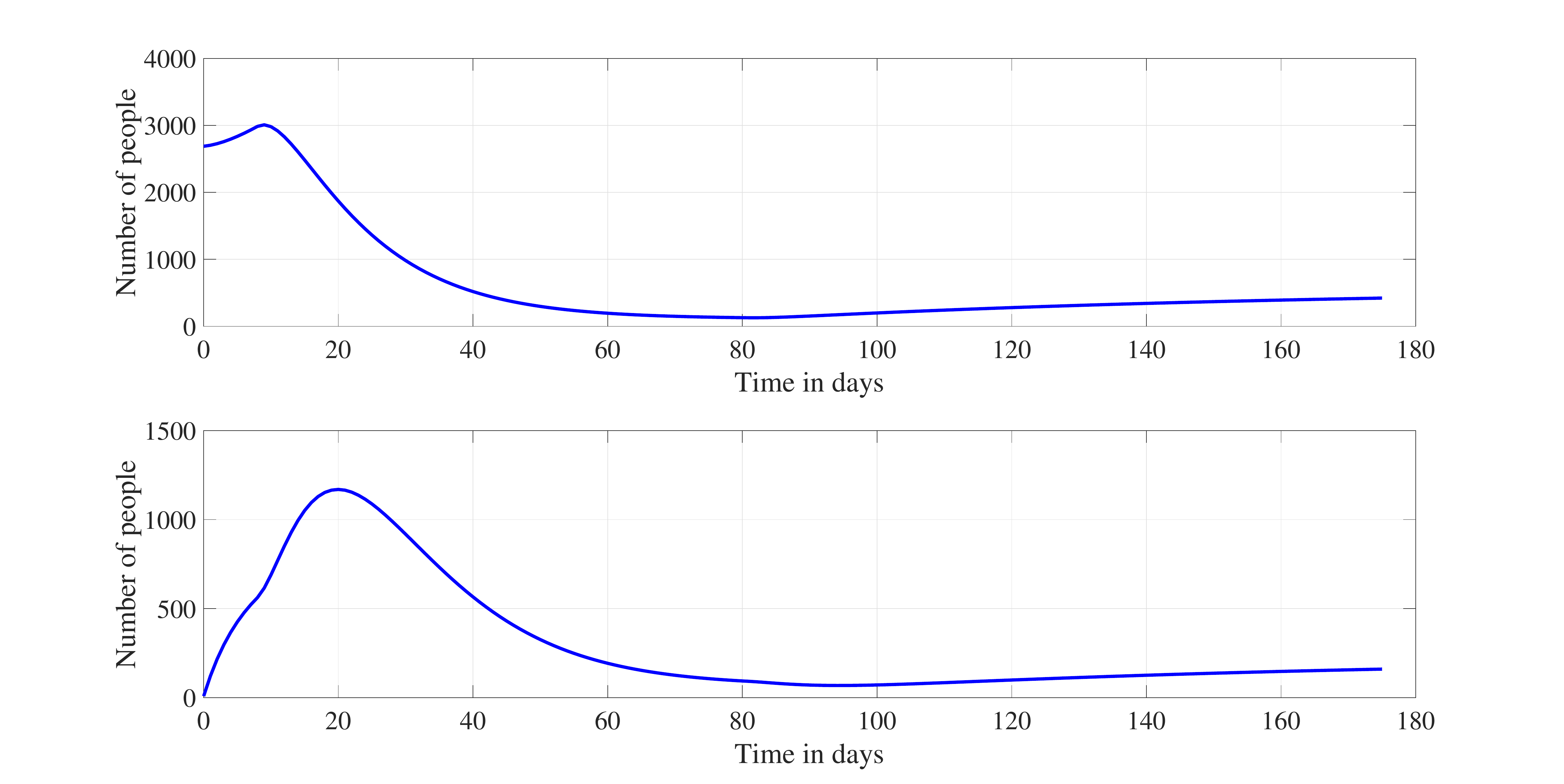} 
   \caption{\small Baseline simulation. The combined exposed ($E_{fc}+E_{pc}$)  (top)  and 
   the combined infected $(I_{fc}+I_{pc})$ (bottom).  175 days from 15 Feb.  until 8 Aug.  2020. 
   $\theta$ changed from $\theta_1=0.01$ to $\theta_2=0.995$ on day 8 and to $\theta_3=0.352$ on day 81.}
   \label{fig:variable-theta-IE}
\end{figure}
\begin{figure}[h!] 
  \hspace{-0.9cm}
  \includegraphics[width=1.15\textwidth]{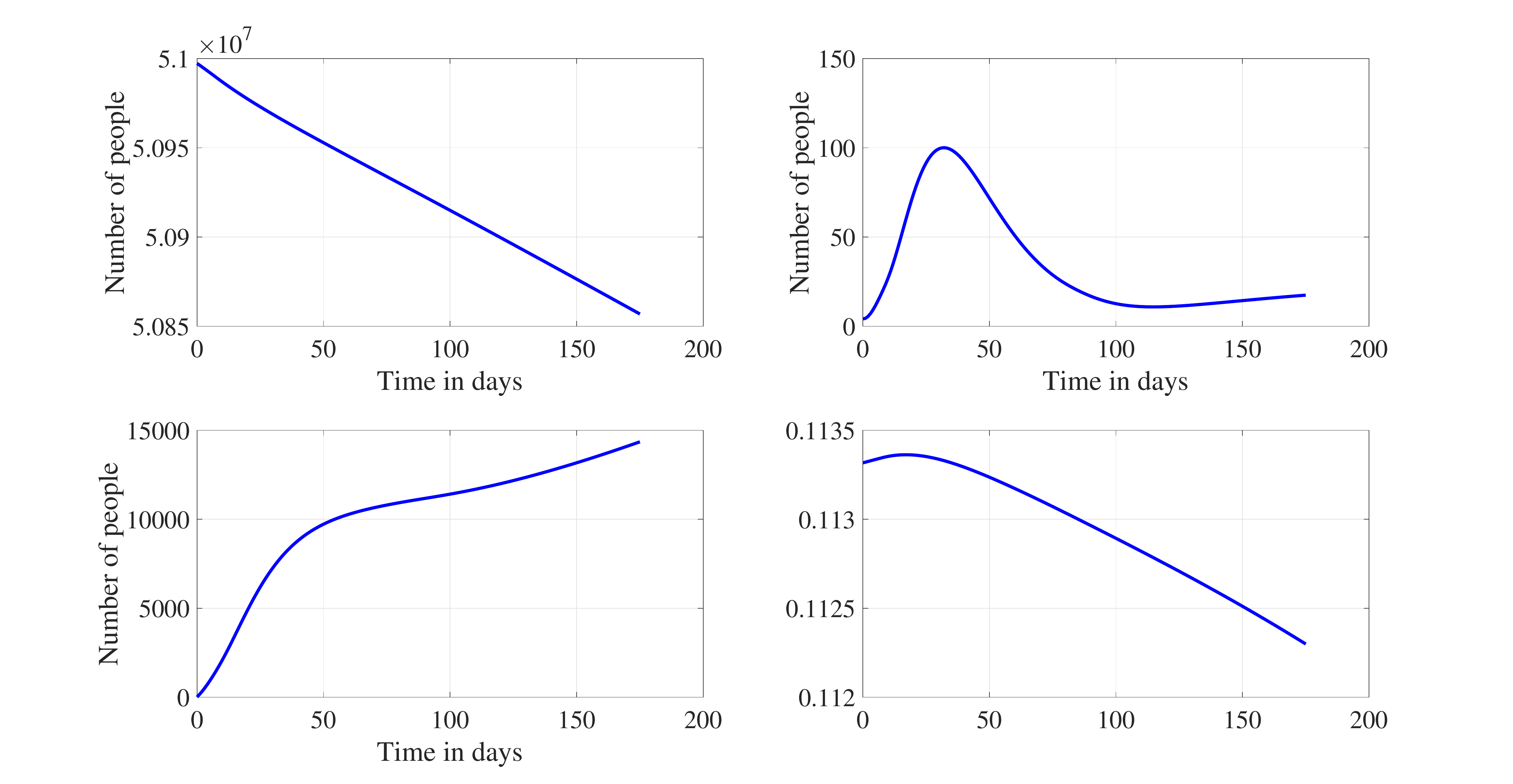} 
   \caption{\small Baseline simulation. The daily numbers of susceptibles $S$ 
   (top left); hospitalized $H$ (top right);
  recovered $R$ (bottom left); and the slight decrease in the disease infectivity $\beta$
   (bottom right).  175 days from 15 Feb.  until 8 Aug.  2020. 
   $\theta$ changed from $\theta_1=0.01$ to $\theta_2=0.995$ on day $8$ and to $\theta_3=0.352$ on day $81$.}
   \label{fig:variable-theta-SHRbeta}
\end{figure}
Fig.~\ref{fig:variable-theta-SHRbeta} depicts the simulation results for the other three 
subpopulations: the number of susceptibles $S$ (top left), the number of 
hospitalized $H$ (top right) and the number of recovered $R$ (bottom left).
The number of hospitalized $H$ should be known precisely  to the authorities, but it was 
not available online for comparison.  Finally, and this is one of the main contributions of this work, 
we show in Fig.~\ref{fig:variable-theta-SHRbeta} (lower right) that the decrease in the contact 
rate $\beta$ in the South Korea case over $175$ days is very small. This is likely due to their success 
in the first $175$ days in controlling the spread of the disease. Running the simulation for $1,000$ 
days shows a decrease to $\beta=0.094$ hence as time increases $\beta$ changes more noticeably. For 
countries with uncontrolled disease spread we expect more significant variation in $\beta$. Therefore, 
considering the infection rate as a dependent variable in that case would lead to more accurate 
model predictions.

To determine the stability of the DFE, we found numerically the eigenvalues of the Jacobian at the DFE using \eqref{eqn54}. We used the baseline parameters, except for 
setting $p_S=\mu N$, $p_{I_{fc}}=\hdots=p_R=0$, $d_I=d_H=0$, so that the population is 
constant, $\beta = 0.094$, and  $\theta=0.995$. Then, the eigenvalues were approximately   
\[
-1.7\,10^{-5}, \; -1.7\, 10^{-5},\; 
  -0.273, \; -0.0188, \;  -0.089,\;
  -0.102, \; -0.158.
\]
Thus, (when $\theta=0.995$) all the eigenvalues are real and negative, so the DFE is stable and attracting.

Hypothetical long runs with $\theta=0.995$  result in  very slow
convergence to the DFE, because the first two eigenvalues are very close to zero.
Indeed, it took over  $1,000$ days to see partial convergence to the limit.

\vskip4pt
As noted above, we found that the DFE looses its stability at $\theta^*\approx 0.034$, which is a bifurcation value 
when the real part of at least one of the eigenvalues becomes positive, and the DFE becomes 
unstable.  Then, theoretically, the EE appears. Long-time simulations with $\theta = 0.352$, 
the value on day $81$ when the directives were relaxed (which is above the critical value),
show, as expected, that the DFE is stable and attracting.

Using the baseline parameters, except for $p_S=\mu N$, $p_{I_{fc}}=\hdots=p_R=0$, $d_I=d_H=0$, $\beta=0.094$,
and $\theta = 0.352$, shows that  the eigenvalues of the Jacobian at the DFE are approximately
\[-1.7\, 10^{-5},\quad  -1.7\, 10^{-5},\; -0.0242, \; 
  -0.0184, \; -0.2384, \; -0.1517, \; -0.1580. 
  \]
Thus, (when $\theta = 0.352$) the DFE  is  stable and  after running the simulation for $1,000$ days
(with constant population of 51 million), we found approximately
\[
S = 50,971,000, \, E_{fc} = 0, \,   E_{pc} = 0,  \,  I_{fc} = 0, \,   I_{pc} = 0, \,   H = 0, \,   R = 29,000.
\]
Since there are no new infections, $R \to 0$ because of natural deaths, but very slowly.
\vskip8pt

  Model simulations with the baseline parameters for $175$ days gives for the
cumulative cases,  deaths, and asymptomatics, respectively (approximately):
\[
CC = 14,058,\quad CD = 322,\quad CA = 5,634,\quad  
\]
and
\[
CFR = 0.022,\quad  IFR = 0.016.
\]

  Simulations with baseline parameters listed in Table~\ref{tab:2} for $1,000$ days (i.e., asymptotically) yields the 
  following projections:
\[
CC = 53,216,\qquad  CD = 1,188, \qquad CA = 23,312, 
\]
\[
CFR = 0.021,\qquad  IFR = 0.015.
\]
We see that the $CFR$ and the $IFR$ are essentially constant. We refer to Subsection \ref{SK-DR}
for additional discussion.

 \vskip4pt 
We remark that at the end of $1,000$ days, $\beta$ decreased to $\beta=0.094$, 
 which indicates that allowing $\beta$
to be a variable may be of importance. On the other hand, since there were 
changes in the various directives following day $175$, our asymptotic numbers 
on day $1,000$ are more of an `if' case.
\vskip4pt

 We next note that in Spain \citep{Lancet20}, about a 
 third of the  individuals who have developed antibodies were asymptomatic, that is, 
 they did not develop any clinical symptoms and were undocumented until  
 tested for antibodies. Using (\ref{eq:CA}), we obtain that the fractions of 
 such cases in South Korea for $175$ days and for $1,000$ days, respectively, are:
\[
\frac{5,634}{14,058}\approx 40.1\%, \qquad \frac{23,312}{53,216}\approx 43.8\%.
\]
 Therefore about $40\%$ of the 
 infections were not detected. This implies that asymptomatic infections contribute silently 
 but substantially to the spread of the disease and more widespread testing for asymptomatics
 may be necessary.

Since our model predictions for the first $175$ days agree well with the data, we have enough confidence 
in our model to proceed  with the study of the  other aspects of the pandemic.
%

\subsection{Effectiveness of control and $\theta$}
\label{SK-thetaImp}
We study next the effectiveness of the control measures in South Korea that are used to contain 
the pandemic, lumped together, and described in the model by the parameter $\theta$. 
To that end, we conduct three hypothetical computer 
experiments showing the model predictions for different control measures. We 
use three other values of $\theta$, representing different responses of the government 
and the population to the pandemic. In these simulations, the $\theta$ values are applied on 
day one and do not change over the $175$ days of simulation. 
In the first simulation, we choose $\theta=0.5$, which means that there are some control measures 
but they are not very strict; in the second simulation $\theta=0.2$; and in the third simulation,  
$\theta=0.1$, that is, very few control measures are in place. All the other 
 system parameters are kept at their baseline values reported in Table~\ref{tab:2}. We note in passing that if any 
 one of these cases was realized, some of the system parameters would have been 
 different, and it is very likely that the outcome would have been worse.

Fig.~\ref{fig:theta0.5-CC-CD} shows that when $\theta=0.5$ the 
number of cumulative cases and deaths on day  $175$ (8 Aug.  2020) 
would increase $1.3$-fold approximately. 
Specifically, on that day the predicted cases would be about $19,739$ 
and deaths about $425$, whereas the real data was at $14,562$ and $304$, respectively.
 \begin{figure}[h!] 
  \hspace{-0.85cm}
  \includegraphics[width=1.10\textwidth]{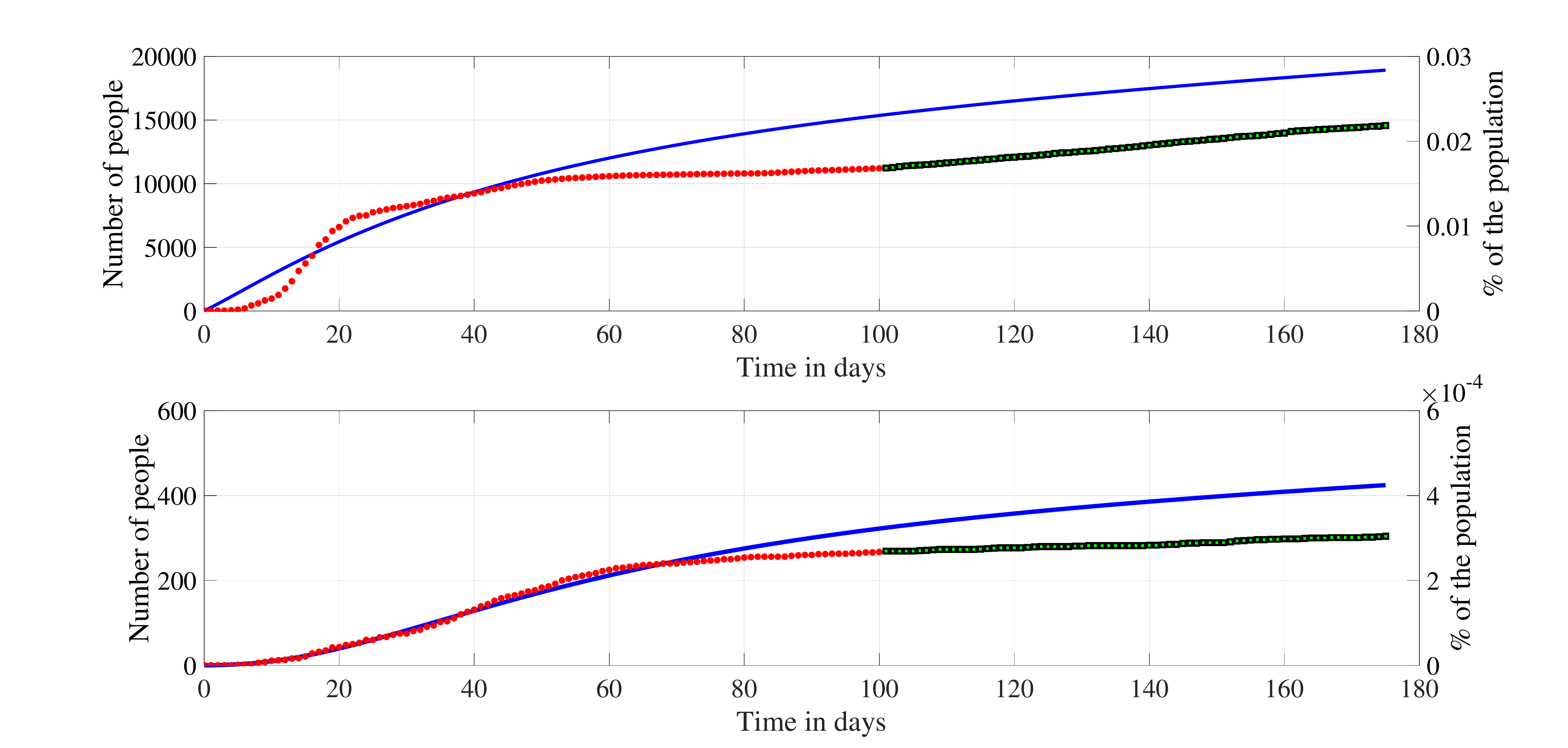} 
  \caption{\small  Simulation with $\theta=0.5$. Model predictions (blue curves) of the cumulative cases (top) and deaths (bottom); 
   over 175 days from 15 Feb. until 8 Aug. 2020. The red filled circles are the $100$ days of field 
   data used in the optimization, the green filled squares depict subsequent data. }
  \label{fig:theta0.5-CC-CD}
\end{figure}

\begin{figure}[h!] 
   \hspace{-0.85cm}
\includegraphics[width=1.10\textwidth]{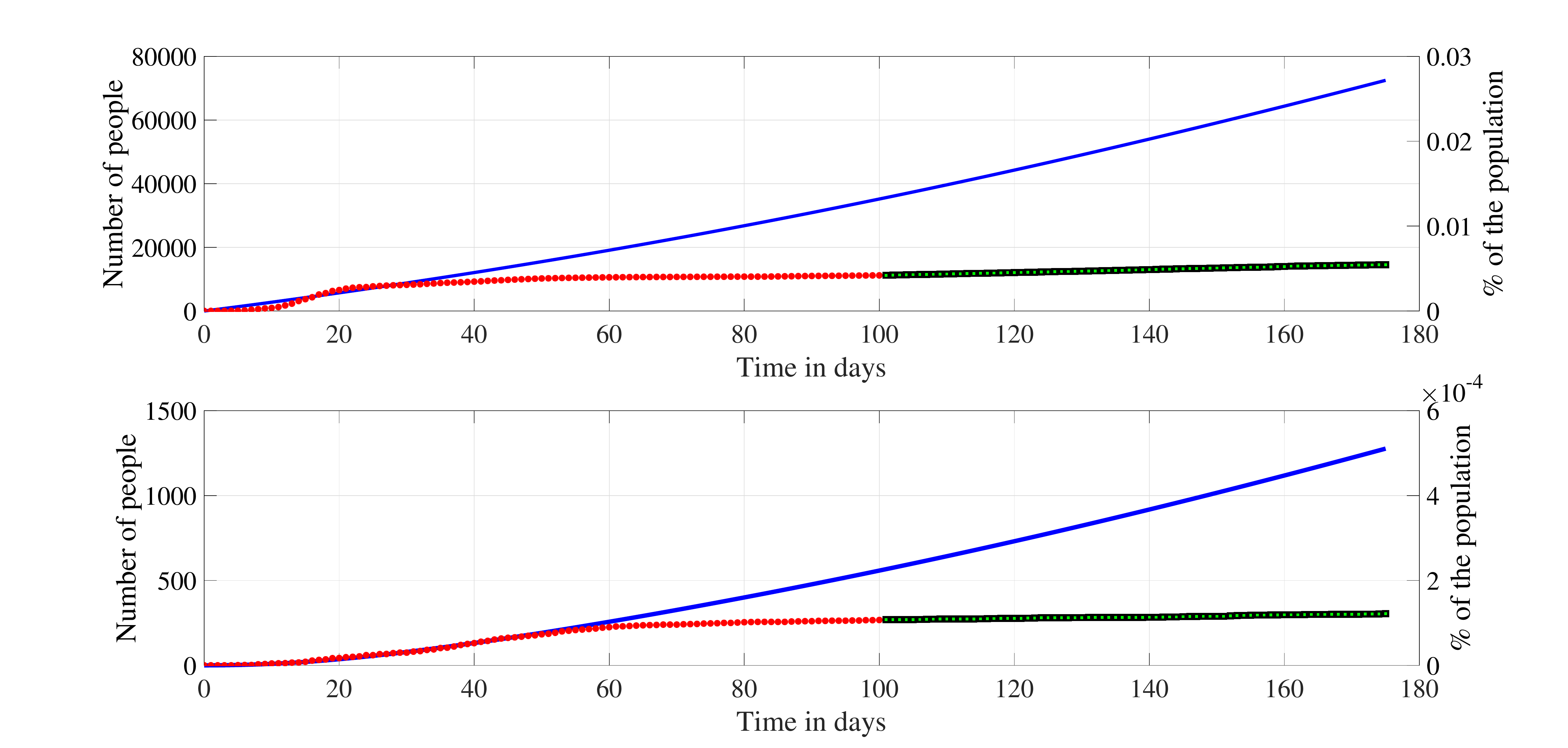} 
   \caption{\small  Simulation with $\theta=0.2$. Cumulative cases (top) and cumulative deaths (bottom). Model predictions (blue curves) of the cumulative cases (top) and deaths (bottom); 
   over 175 days from 15 Feb. until 8 Aug. 2020. The red filled circles are the $100$ days of field 
   data used in the optimization, the green filled squares depict subsequent data. }
   \label{fig:theta0.2-CC-CD}
\end{figure}
Next, Fig.~\ref{fig:theta0.2-CC-CD} shows that when $\theta=0.2$, the 
number of cumulative cases would be more than 5-fold and the number of deaths on day $175$ would be more than 
$4$-fold. Specifically, the model predicts that on day $175$ there would be about 
$73,288$ cases and $1,276$ deaths.

Finally, it is seen in Fig.~\ref{fig:theta0.1-CC-CD} that if $\theta=0.1$ the 
number of cumulative cases on day $175$ would have jumped more than 
$10$-fold and the  number of cumulative deaths more than  $7$-fold. 
Specifically, on day $175$, the model predicts about $158,613$ cases and $2,387$ 
deaths. 

Furthermore, it is common sense to expect that less control measures would lead to
worse outcomes and this is confirmed by our model predictions showing substantially worse consequences with relaxed control measures. Indeed, higher values of $\theta$ are quite effective in avoiding large scale  
disruptions of the health care system, and all other state and  economic systems 
because of the pandemic.

\begin{figure}[h!] 
   \hspace{-0.85cm}
   \includegraphics[width=1.10\textwidth]{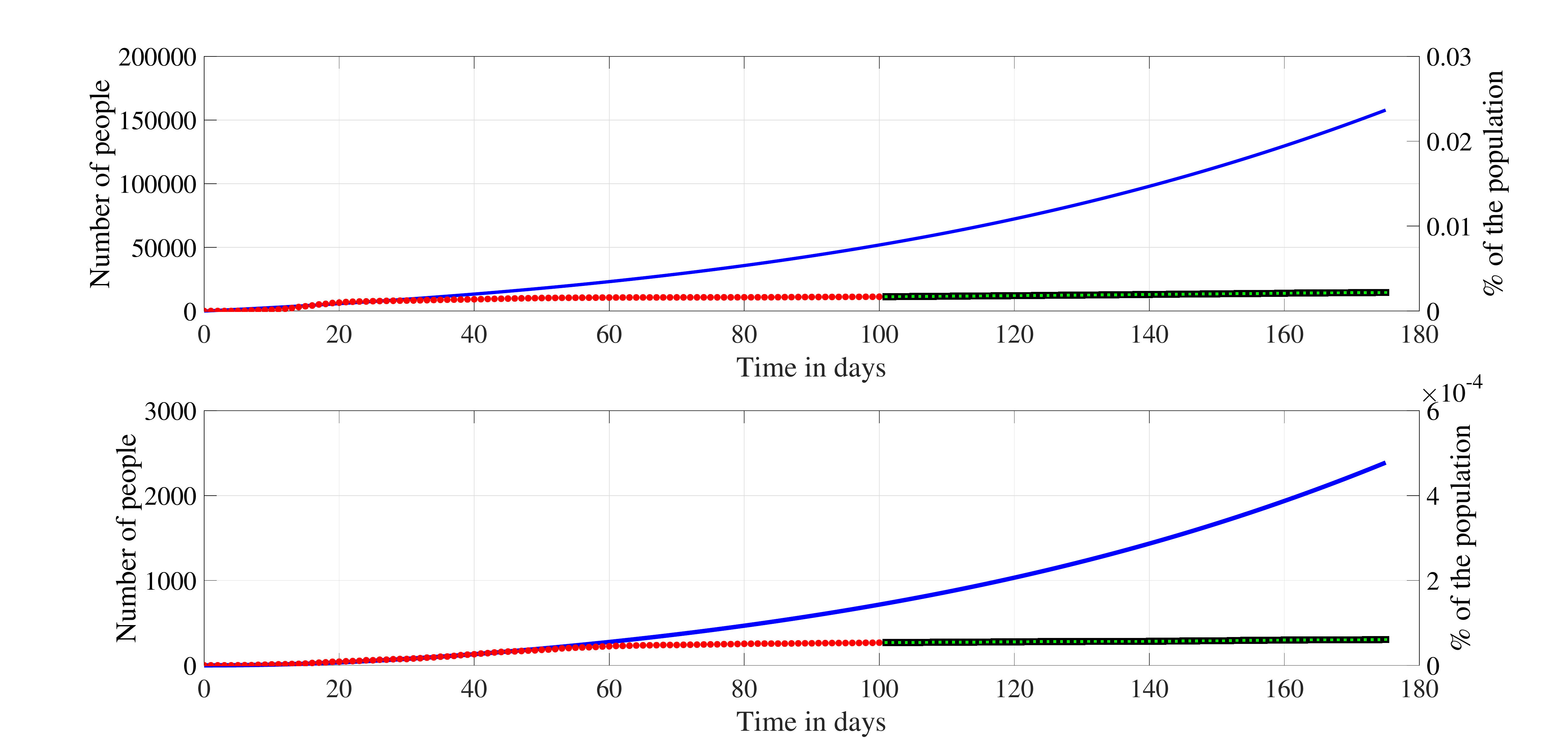} 
   \caption{\small  Simulation with $\theta=0.1$. Cumulative cases (top) and cumulative deaths (bottom).}
   \label{fig:theta0.1-CC-CD}
\end{figure}

\subsection{The COVID-19 death rates in South Korea}
\label{SK-DR}
Since \textbf{there is a considerable discussion} in the literature and especially in the media,
about the COVID-19 death rates, we use the baseline simulations to determine
two death rates, (see, e.g.,\citep{wikiterms}). The first, 
the {\it case fatality rate} (CFR), $\mu_{cov19}^{**}(t)$, is the ratio of the total deaths 
caused by the disease to all those who have been diagnosed (or documented) with the disease, 
and this includes the current infectives, hospitalized, and those who recovered or died, 
 and is given by
\[
CFR=\mu^{**}_{cov19}(t)=\frac{CD(t)}{CC(t)}.
\]
Here, $CD(t)$ is the cumulative number of deaths, (\ref{eqn218}),
and $CC(t)$, (\ref{eqn215}), is the cumulative number of infected, up to time $t$.
The second, the {\it infection fatality rate} (IFR), $\mu^*_{cov19}(t)$, is the ratio
of the deaths to all those who had the virus, including the asymptomatics, and is given as
\[
IFR= \mu^*_{cov19}(t)=\frac{CD(t)}{CC(t)+CA(t)}.
\]
\begin{figure}[h!] 
  \hspace{-1.4cm}
  \includegraphics[width=1.20\textwidth]{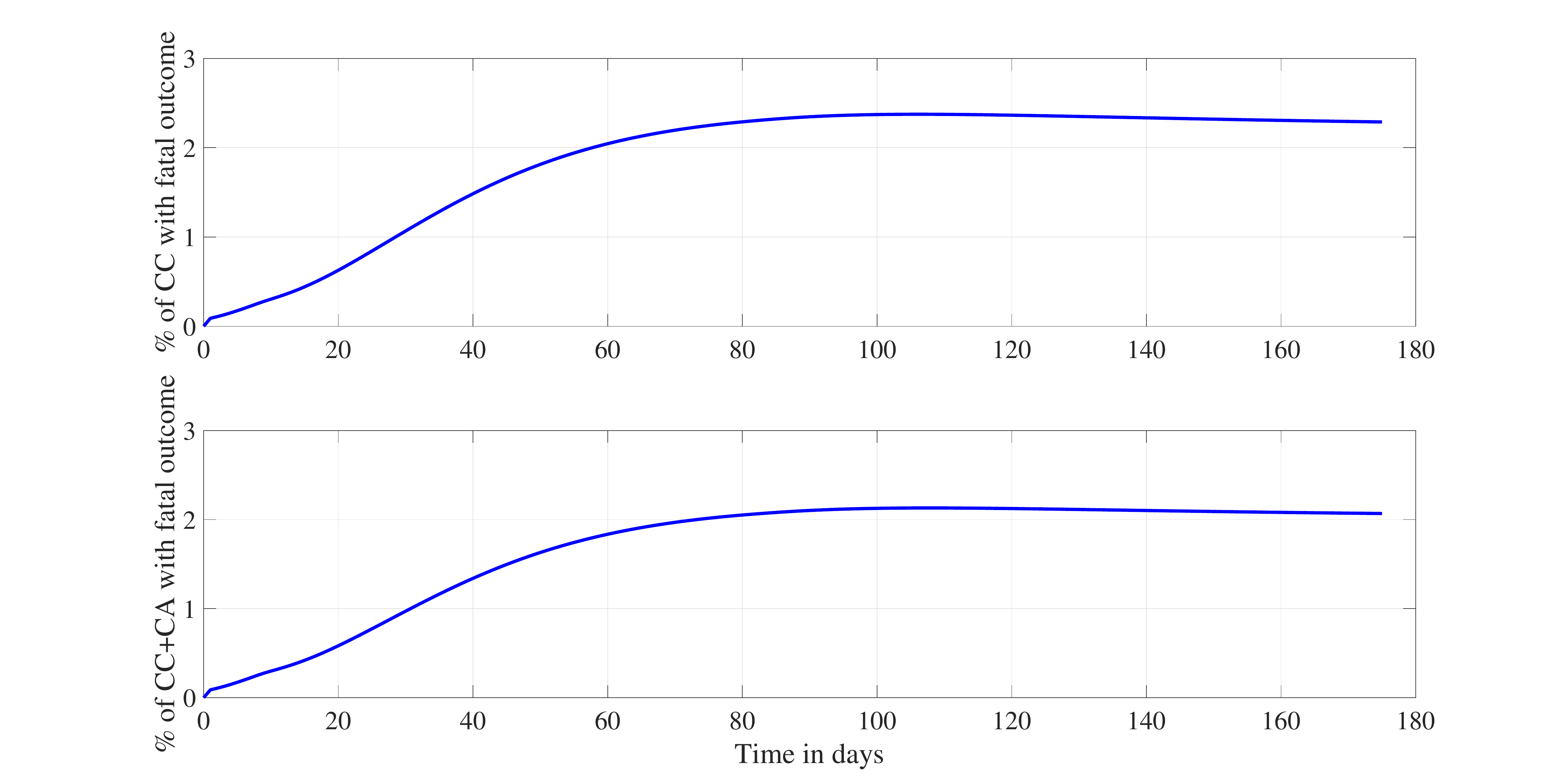}
   \caption{\small The death rates (CFR)=$\mu^{**}_{cov19}$, (top) and 
   (IFR)=$\mu^{*}_{cov19}$ (bottom)
   as functions of time in the baseline simulation.
   }
\label{fig:mucov19}
\end{figure}

Fig.~\ref{fig:mucov19} shows that
on day $175$ ($8$ Aug 2020) the CFR was approximately $2.16\%$ and the IFR was $1.57\%$. 
According to the official  KCDC \citep{KCDC}, 
on that day there were $304$ deaths and $14,562$ cases, hence
the `official' $CFR=2.09\%$ was slightly less than the model prediction.
Whereas the $CFR=\mu^{**}_{cov19}$ can be 
found in official publication, the $IFR=\mu^*_{cov19}$ cannot, since it includes those 
who were never officially counted as having the virus, and usually one has to revert 
to different ways of estimating it. Thus, our model provides a means to estimate this 
death rate that has been used (often loosely) in the media. 
%

\section{Conclusions and further research}
\label{sec-con}

This work presents a new compartmental model of the SEIR-type
for the dynamics of the COVID-19 pandemic, in particular allowing  the  
assessment of the effectiveness of the related disease control and containment measures. 
The model's novelty is two-fold. First, the populations are split
into those who follow fully all the control measures and those who follow only partially, 
such as essential workers who need to be mobile and cannot socially distance, or those who choose to only partially follow 
the directives. The split is controlled by the parameter $\theta$ that also takes indirectly 
into account the `quality' of these measures. Second, the 
infection contact rate coefficient  $\beta(t)$ is a dependent variable, the dynamics 
of which is governed by a differential inclusion, (\ref{eqn210}). 
This takes into account the intrinsic variability of $\beta$ due to various causes of infection rate 
changes, such as changes in the population behavior or improvements in treatments. 

Since the model is complex and includes a differential inclusion, the existence of its 
unique solution on every finite time interval is established by a nonstandard proof, 
using results from convex analysis. The positivity of the solutions is also established.
Moreover, when the data is measurable with respect to any random variables, then so are the 
solutions. 

Several baseline simulations, in MATLAB, of the disease dynamics were conducted for South 
Korea. The optimal model parameters were obtained by an ($\ell^1$ based) optimization 
routine fitting the model solutions with the data of both cumulative and daily cases and deaths. 
In particular, the parameter $\theta$ that controls the split between fully and partially 
compliant populations is found to be $0.010$  for the first eight days, on day eight (22 Feb. 
2020) when a lock-down and strict control measures were implemented, it jumped to $0.995$ 
and as the control measures started to ease on day $81$ (6 May 2020) it dropped to $0.352$.
The baseline simulations show that the model captures the disease dynamics very well 
when compared to the data for daily cases and deaths  (with seven day moving averages)
and cumulative cases and deaths. 
Furthermore, it provides additional information about the pandemic dynamics that 
are difficult to observe in the field, such as $40\%$ asymptomatics.

One of the model conclusions, when applied 
to South Korea, is that their success in controlling and reducing the infection rates is 
related to having a period of $73$ days in which $\theta=0.995$, which indicates that 
the majority of the population adhered to the national directives and disease control measures, which in turn were consistent and clearly applied. 

In addition, our model predicts substantially worse consequences with 
relaxed control measures which makes sense as it is common sense to expect that less control measures would lead to worse outcomes. Indeed, higher values of $\theta$ are quite effective to avoid large scale  disruptions of the health care system, and all other state and  economic systems 
because of the pandemic. This, however, is in the absence of vaccinations.

To gain further insight into the model predictions, we studied the system's equilibrium points. 
Due to the complexity of the model, instead of finding a basic stability number $\mathcal{R}_C$, 
we used the the eigenvalues of the system Jacobian. 
 We found that the disease-free equilibrium (DFE) is asymptotically stable 
when $\theta=0.995$ and becomes unstable when  $\theta^*\approx 0.034$, and since we used 
$\theta=0.352$ from day $81$ onward, the DFE is stable and attracting. 

We also used the baseline simulations to compute the death rates and found that the case fatality 
rate was $CFR=2.16\%$ and the infection fatality rate was $IFR=1.57\%$ on day $175$, as shown in 
Fig.~\ref{fig:mucov19}. The simulations for $1,000$ days show that the rates are essentially the same.

The next topic we investigated was the effectiveness of the disease containment measures controlled by the values of $\theta$. Simulations with $\theta=0.5, 0.2, 0.1$ show that not using sufficient control 
measures would have caused a substantial increase in the number of cases and deaths. Indeed, 
for $\theta=0.1$ there would be more than a $10$-fold increase in cases and $7$-fold 
increase in fatalities. These results were obtained without taking into account large-scale 
vaccination.

An important result that the model provided is that the asymptomatics, 
those who were carry the virus without symptoms, constituted about $40\%$ 
of the cases. This is very important information that is very difficult and costly to 
obtain in the field.

We also performed a partial sensitivity study of the models dependence on $\theta$, and 
the infection rates of the exposed $(\gamma_{fc}, \gamma_{pc})$. We found that
the model is sensitive to these parameters. Thus, to obtain reliable predictions, 
these model parameters need to be estimated accurately. However, since we 
did not conduct a full sensitivity analysis with respect to the full set of 38 model parameters, 
we did not present these impartial results in this work.
\vskip6pt

We now describe some of the unresolved issues that may be of interest for further
study:

(i) Add the effects of large-scale vaccination to the model and study the model predictions. This will
bring the model up-to-date and entails modest modifications of the model.

(ii) Modify the model to take into account new variants of the SARS-CoV-2. 
This entails major modifications of the model.

(iii) Replace the jump of $\theta$ to $0.352$ after day $81$ when the lock-down control 
measures started to ease, with an appropriate function of time, $\theta=\theta(t)$.

(iv)  Compare the model prediction of the hospitalized $H(t)$, (Fig. 5 
top right), and the field data.

(v) Introduce more $\theta$-like parameters to allow for the 
separate assessment of the control measures, such as wearing a face mask in public, 
washing hands often, keeping distance in public spaces, and
following the instructions. However, such an expansion may result in a more
complex model that may be more difficult to work with.

From the mathematical point of view, it may be of interest to:

Establish rigorously the bifurcation property in $\theta$; the convergence of the solutions
to the DFE when it is stable and attracting, and the convergence to the EE when it exists.
Determine that there are no other equilibrium points and analyze in more depth the properties 
of the two states. Study more thoroughly the differential inclusion for $\beta$, 
consider alternative formulations based on more general biological information about the 
disease based on published observations.  Find the optimal regularity of the solutions.
Study alternative objective functionals by assigning different weights for the number 
of cases and deaths and using an $\ell^2$-type minimization.
\vskip6pt
We conclude that the model has been sufficiently validated for the pandemic in South Korea.
It both provides insight and allows for `mathematical experiments.' We plan to use it for 
other countries and states for which reliable data can be found.

\begin{acknowledgements}
We would like to thank the referee and the editor for their thorough and helpful comments 
that improved the work. 
\end{acknowledgements}

%
\section*{Conflicts of interest}
The authors did not receive support from any organization for the submitted work.
The authors have no conflicts of interest to declare that are relevant to the content of this article.

\bibliographystyle{spbasic}      
\bibliography{references}   

\newpage
\begin{appendices}
\section{Jacobian of the system}
\label{appendix}
This appendix constructs an expression for the system's Jacobian assuming that $\beta=const.$

%
For this purpose, we write equations 
(\ref{eqn21})--(\ref{eqn27}) using the notation introduced in \eqref{eq:hat}  as
\begin{equation*}
\label{eq:system}
\frac{d\bx(t)}{dt}=\bff(t),
\end{equation*}
where the components of $\bff(t)$ are given by
\begin{align*}
f_{1}(t) & =   p_{S}- \Gamma S -\mu S,   \\
f_{2}(t)& = p_{Efc}+ \theta \Gamma S  - \widehat{\gamma}_{fc} E_{fc} 
 +\frac{\gamma_{E}}{N} E_{fc}E_{pc}, \\
f_{3}(t) & = p_{Epc}+ (1-\theta)\Gamma S  - \widehat{\gamma}_{pc} E_{pc}  
 -\frac{\gamma_{E} }{N}E_{fc}E_{pc}, \\
f_{4}(t) & = p_{Ifc}+  \gamma_{fc}E_{fc} +\gamma^{-}E_{pc}
- \widehat{ \delta} I_{fc}  +\frac{\gamma_{I}}N  I_{fc}I_{pc} ,    \\
f_{5}(t) & =  p_{Ipc}+  \gamma_{pc}E_{pc} +\gamma^{+}E_{fc}
- \widehat{\delta} I_{pc}  -\frac{\gamma_{I}}N  I_{fc}I_{pc},   \\
f_{6}(t) & = p_{H}+ \delta (I_{fc} + I_{pc})- \widehat{\sigma}_H H,    \\
f_{7}(t) & =  p_{R} + \sigma_{E}(E_{fc} + E_{pc})+ \sigma_{I}(I_{fc}
+ I_{pc})+\sigma_H H -\mu R. 
\end{align*}
The Jacobian matrix of the system is 
\begin{equation*}
\label{A1}
J(\bx)=\left(\begin{array}{ccc}
J_{11}& \dots& J_{17}\\
\vdots&\ddots&\vdots\\
J_{71}& \dots&J_{77}
\end{array}
\right),
\end{equation*}
where $J_{ij}=\partial f_{i}/\partial x_{j}$, $1\leq i, j,\leq 7$. 
To compute $J(\bx)$, we note that
\begin{align*}
\frac{d \Gamma}{dS }&=-\frac 1N \Gamma, \\
\frac{d \Gamma}{dE_{fc} }&=-\frac 1N (
\Gamma -\beta \epsilon_{Efc}),&& \frac{d \Gamma}{dE_{pc} }=-\frac 1N (
\Gamma -\beta \epsilon_{Epc}),
\\
\frac{d \Gamma}{d I_{fc} }&=-\frac 1N (\Gamma-\beta \epsilon_{Ifc}), 
&& \frac{d \Gamma}{dI_{pc} }=-\frac 1N (
\Gamma -\beta \epsilon_{Ipc}),
\\
\quad \frac{d \Gamma}{dH }&=-\frac 1N (
\Gamma -\beta \epsilon_{H}),&&
\frac{d \Gamma}{dR }=-\frac 1N \Gamma.
\end{align*}
Then, straightforward and rather tedious manipulations yield,
\begin{align*}
J_{11}&=-\Gamma \left(1-\frac SN\right)-\mu,\, J_{12}=  \frac SN(\Gamma-\beta \epsilon_{Efc}),\, J_{13}= \frac SN(\Gamma-\beta \epsilon_{Epc}),
\\
J_{14}&= \frac SN(\Gamma-\beta \epsilon_{Ifc}),\, J_{15}= \frac SN(\Gamma-\beta \epsilon_{Ipc}),
\,
J_{16}= \frac SN(\Gamma-\beta \epsilon_{H}),\, J_{17}= \frac SN\Gamma,
\\
J_{21}&=\theta \Gamma \left(1-\frac SN\right),\,J_{22}=  -\theta \frac SN(\Gamma-\beta \epsilon_{Efc})
-\widehat{\gamma}_{fc}+\gamma_{E}\frac {E_{pc}}N,
\\
J_{23}&= -\theta\frac SN(\Gamma-\beta \epsilon_{Epc})+\gamma_{E}\frac {E_{fc}}N,\,
J_{24}=-\theta\frac SN(\Gamma-\beta \epsilon_{Ifc}),
\\
J_{25}&= -\theta\frac SN(\Gamma-\beta \epsilon_{Ipc}),\,
J_{26}= -\theta \frac SN(\Gamma-\beta \epsilon_{H}),\,  J_{27}= -\theta\frac SN\Gamma -\dfrac{\gamma_E}{N^2}E_{fc}E_{pc},
\\
J_{31}&=(1-\theta) \Gamma \left(1-\frac SN\right),\, J_{32}=  -(1-\theta) \frac SN(\Gamma-\beta \epsilon_{Efc})
-\gamma_{E}\frac {E_{pc}}N,
\\
J_{33}&= -(1-\theta)\frac SN(\Gamma-\beta \epsilon_{Epc})-\widehat{\gamma}_{pc} -\gamma_{E}\frac {E_{fc}}N,\,
J_{34}=-(1-\theta)\frac SN(\Gamma-\beta \epsilon_{Ifc}),
\\
\quad J_{35}&= -(1-\theta)\frac SN(\Gamma-\beta \epsilon_{Ipc}),\,
J_{36}= -(1-\theta) \frac SN(\Gamma-\beta \epsilon_{H}),\\
J_{37}&=-(1-\theta)\frac SN\Gamma+\dfrac{\gamma_E}{N^2}E_{fc}E_{pc},
\\
J_{41}&=-\gamma_{I}\frac{I_{fc}I_{pc}}{N^{2}}
,\, J_{42}={\gamma}_{fc}-\gamma_{I}\frac{I_{fc}I_{pc}}{N^{2}},\,
J_{43}= {\gamma}^{-}-\gamma_{I}\frac{I_{fc}I_{pc}}{N^{2}},
\\
J_{44}&= -\widehat{\delta}+ \gamma_{I}\frac{I_{pc}}N -\gamma_{I}\frac{I_{fc}I_{pc}}{N^{2}},\,J_{45}=  \gamma_{I}\frac{I_{fc}}N -\gamma_{I}\frac{I_{fc}I_{pc}}{N^{2}},
\\
J_{46}&=-\gamma_{I}\frac{I_{fc}I_{pc}}{N^{2}},\,
J_{47}=-\gamma_{I}\frac{I_{fc}I_{pc}}{N^{2}}
\\
J_{51}&=\gamma_{I}\frac{I_{fc}I_{pc}}{N^{2}}
,\, J_{52}={\gamma}^{+}+\gamma_{I}\frac{I_{fc}I_{pc}}{N^{2}},\,
J_{53}= {\gamma}_{pc}+\gamma_{I}\frac{I_{fc}I_{pc}}{N^{2}},
\\
J_{54}&= -\gamma_{I}\frac{I_{pc}}N +\gamma_{I}\frac{I_{fc}I_{pc}}{N^{2}},\, J_{55}=-\widehat{\delta}- \gamma_{I}\frac{I_{fc}}N -\gamma_{I}\frac{I_{fc}I_{pc}}{N^{2}}, \\
J_{56}&=\gamma_{I}\frac{I_{fc}I_{pc}}{N^{2}},\,
J_{57}=\gamma_{I}\frac{I_{fc}I_{pc}}{N^{2}},
\\
J_{61}&=0,\, J_{62}=0,\quad J_{63}= 0,\, J_{64}= \delta,\, J_{65}=\delta,\,
J_{66}=-\widehat{\sigma}_{H},\,  J_{67}=0,
\\
J_{71}&=0,\quad J_{72}=\sigma_{E},\quad J_{73}= \sigma_{E},\, J_{74}= \sigma_{I},
\quad J_{75}=\sigma_{I},\,
J_{76}=\sigma_{H},\,  J_{77}=-\mu.
\end{align*}

\end{appendices}

\end{document}